\documentclass[journal]{IEEEtran}  

\usepackage{graphics} 
\usepackage{mathrsfs}
\usepackage{times} 
\usepackage{amsmath} 
\usepackage{amssymb}  

\usepackage{multirow}
\usepackage{amsthm}
\usepackage{mathtools}
\usepackage{commath}
\usepackage{graphicx}
\usepackage{algorithm}
\usepackage{algpseudocode}
\usepackage[bookmarks=false]{hyperref}

\usepackage[T1]{fontenc} 

\newtheorem{theorem}{Theorem}
\newtheorem{lemma}{Lemma}
\newtheorem{corollary}{Corollary}

\theoremstyle{definition}
\newtheorem{definition}{Definition}
\newtheorem{assumption}{Assumption}

\def\equationautorefname~#1\null{(#1)\null}
\def\sectionautorefname~#1\null{Section~#1\null}
\def\subsectionautorefname~#1\null{Section~#1\null}
\def\algorithmautorefname~#1\null{Algorithm~#1\null}
\def\theoremautorefname~#1\null{Theorem~#1\null}
\def\lemmaautorefname~#1\null{Lemma~#1\null}
\def\corollaryautorefname~#1\null{Corollary~#1\null}
\def\remarkautorefname~#1\null{Remark~#1\null}
\def\definitionautorefname~#1\null{Definition~#1\null}
\def\assumptionautorefname~#1\null{Assumption~#1\null}

\begin{document}

\title{A Spectral Perspective on Stochastic Control Barrier Functions}
\author{
Inkyu Jang,
Chams E. Mballo,
Claire J. Tomlin,
and
H. Jin Kim
\thanks{Inkyu Jang was with the Department of Aerospace Engineering, Seoul National University, Seoul 08826, Republic of Korea. He is now with the Department of Electrical Engineering and Computer Sciences, University of California, Berkeley, CA 94720 USA (e-mail: \href{mailto:janginkyu@berkeley.edu}{janginkyu@berkeley.edu}).}
\thanks{Chams E. Mballo and Claire J. Tomlin are with the Department of Electrical Engineering and Computer Sciences, University of California, Berkeley, CA 94720 USA (e-mail: \href{mailto:cmballo@berkeley.edu}{cmballo@berkeley.edu}; \href{mailto:tomlin@eecs.berkeley.edu}{tomlin@eecs.berkeley.edu}).}
\thanks{H. Jin Kim is with the Department of Aerospace Engineering and Automation and Systems Research Institute, Seoul National University, Seoul, 08826 Republic of Korea (e-mail: \href{mailto:hjinkim@snu.ac.kr}{hjinkim@snu.ac.kr}).}
\thanks{This work has been submitted to the IEEE for possible publication. Copyright may be transferred without notice, after which this version may no longer be accessible.} 
}

\maketitle

\begin{abstract}
Stochastic control barrier functions (SCBFs) provide a safety-critical control framework for systems subject to stochastic disturbances by bounding the probability of remaining within a safe set.
However, finding a valid SCBF that explicitly reflects the long-horizon behavior of the true safety probability remains a challenge. This paper addresses this issue by adopting a spectral perspective, utilizing the linear operator that governs the evolution of the closed-loop system's safety probability. We find that the dominant eigenpair of this Koopman-like operator encodes fundamental safety information of the stochastic system. The dominant eigenfunction is a natural and valid SCBF, with values that characterize relative asymptotic safety across initial states, while the dominant eigenvalue indicates the global asymptotic rate at which the safety probability decays. A practical algorithm for computation is proposed, termed power-policy iteration, which jointly estimates the dominant eigenpair while updating the backup policy with the aim of improving long-term safety. The method is validated using simulation experiments on safety-critical dynamics models.
\end{abstract}


\section{Introduction}
\subsection{Motivation}

\IEEEPARstart{S}{afety} is a fundamental requirement in the design and deployment of cyber-physical systems. In such systems, this requirement is enforced by prescribing state constraints that must be satisfied at all times, which naturally leads to a set-invariance-based characterization of safety. Control barrier functions (CBFs) provide a principled means of enforcing these constraints by certifying forward invariance of a prescribed safe set~\cite{ames2019control}. If the system is initialized in this safe set, the trajectory is guaranteed to remain within it indefinitely under a CBF-based safety controller.

In many practical applications, however, system dynamics are inherently stochastic due to disturbances and uncertainties such as unmodeled forces and sensing errors. Under such stochastic disturbances, particularly when they are unbounded, one cannot guarantee that the system remains indefinitely in the safe set. Consequently, even under a well-designed safe controller, constraint violations may occur with nonzero probability. This has led to the development of stochastic control barrier functions (SCBFs), which extend classical CBFs to stochastic dynamics. In contrast to deterministic invariance, SCBFs encode safety in probabilistic terms. Specifically, SCBFs characterize safety by lower-bounding the \textit{safety probability}, i.e., the probability of remaining in the safe set over a given time horizon \cite{steinhardt2012finite, santoyo2019verification, gao2020computing, singletary2022safe}.

CBFs, including SCBFs, are often used as safety monitors that act as supervisory filters for a nominal controller \cite{wabersich2023data}. At each control update, the safety monitor evaluates the current barrier function value and its predicted rate of change under the proposed input to determine whether the system is being driven toward the unsafe region at an excessively fast rate. If so, the monitor intervenes by modifying, often minimally \cite{cbfqp}, the proposed input to obtain a safe alternative.
The alternative input is governed by a \textit{backup policy}, serving as a baseline safe strategy that ensures the existence of such a safe alternative within the control bounds.
For this reason, this safety monitor is often referred to as a \textit{safety filter}.

In this supervisory role, a barrier function is most useful when its value can be directly interpreted as an indicator of system safety. 
For instance, in deterministic systems, safety is determined solely by whether the state lies inside or outside the invariant set; in this setting, the sign of the CBF suffices to distinguish between safe and unsafe states.
In stochastic systems, however, safety is no longer binary, since constraint violation may occur with probability strictly between zero and one.
Accordingly, the value of an SCBF should instead convey how safe a given state is in probabilistic terms, especially because the SCBF is designed to provide a lower bound on the safety probability.

However, mere validity as an SCBF does not imply that the function serves as a consistent proxy for safety. Existing approaches view SCBFs only indirectly through the associated feasibility conditions, without explicitly calibrating the values to the safety probability, which can lead to loose bounds that obscure the true safety margin.
Moreover, because the values of such SCBFs do not necessarily scale monotonically with the actual safety probability, it becomes difficult to interpret them as a consistent safety metric or to meaningfully compare the risk levels across states.

\subsection{Summary of Contributions}

Given the motivation above, this paper addresses the following research question:
\begin{quote}
    How does one characterize an SCBF that meaningfully captures the safety properties of a given stochastic system, and how can such an SCBF be computed?
\end{quote}
To answer this, we first formalize a time-dependent Koopman-like linear operator that governs the safety properties of the closed-loop stochastic system modeled as an It\^o diffusion process.
In contrast to conventional Koopman operators, this operator propagates observables only along trajectories that have remained within the safe set, by assigning zero contributions to trajectories that have exited the safe set.
Thus, its action on the indicator function for the safe set yields the safety probability.
This operator naturally arises from the dynamic programming principle that describes the evolution of safety probability of the system.
The dominant (biggest in absolute value) eigenvalue of this operator describes the global rate at which the closed-loop system's safety probability decays over long time horizons, and therefore is an indicator for the system's overall safety.
The corresponding eigenfunction, called the dominant eigenfunction, characterizes the relative safety across initial states along closed-loop trajectories, in terms of asymptotic behavior of the safety probability.

We also find that the dominant eigenfunction is a valid and natural choice of SCBF for the system, as it is explicitly connected to the safety probability and therefore is interpretable and meaningful.
We show that, under the stated assumptions, the dominant eigenfunction is uniquely determined up to scale, depending on the chosen backup policy.
In contrast to general SCBFs, whose validity alone guarantees only a lower bound on the safety probability, the proposed dominant eigenfunction provides \textit{both} an SCBF-based lower bound and an asymptotic characterization of the actual safety probability.

We propose a practical computational algorithm called \textit{power-policy iteration} that computes jointly the SCBF (i.e., the dominant eigenpair) and the backup policy.
Power-policy iteration adds a policy update step to the classic power iteration algorithm, with the aim of improving the long-term safety of the backup policy along the course of iteration.
The numerical examples show that the iteration exhibits empirical convergence and improvements in the asymptotic decay rate of the safety probability.
The converged decay rates are consistent with Monte Carlo estimates of the safety probability under the obtained backup policies.

\subsection{Related Work} \label{sec: related work}

\subsubsection{Control Barrier Functions}
The CBF framework has been a popular tool for enforcing invariance of a set over an infinite time horizon in the deterministic setting \cite{ames2019control}. 
By mapping safety constraints to conditions on the time derivative of a scalar Lyapunov-like function, the framework enables computation of safety filters, often implemented as quadratic programs (QPs), called CBF-QP \cite{cbfqp}.
This framework has been successfully applied across a wide range of domains, including, but not limited to, robotics \cite{ferraguti2022safety, jang2023safe2} and aerospace \cite{molnar2025collision, shahraki2025spacecraft}.

SCBFs are extensions of CBFs to systems subject to stochastic disturbances \cite{prajna2007framework, santoyo2019verification, nishimura2024control}.
In general stochastic systems under stochastic disturbances with unbounded support, it is not possible in general to achieve a deterministic or almost-sure safety guarantee as in classical CBFs \cite{so2025comment}.
Instead, many SCBF approaches provide a lower bound on the probability that the system remains within the safe set.
Probabilistic safety guarantees are often established using martingale inequalities or their variants.
Such approaches have been developed for discrete-time systems \cite{jagtap2021formal, cosner2023robust, cosner2024bounding, salamati2024data, zikelic2023learning}, continuous-time systems \cite{prajna2007framework, anand2019verification, santoyo2019verification, yaghoubi2021risk}, and both \cite{steinhardt2012finite, santoyo2021barrier, jang2025upper, liu2026safety}.
Some works also explored CBF variants for uncertain or stochastic systems with respect to different risk measures beyond mere probability, for example, exponential utility \cite{lederer2025risk}, conditional value-at-risk \cite{ahmadi2022risk, kishida2024risk}, or more generally, coherent risk measures \cite{singletary2022safe}.

Despite the theoretical utility of the CBF framework and its extensions, finding a valid barrier function remains a challenge, especially for systems with input bounds.
The (S)CBF conditions are described as a set of differential inequalities that need to be satisfied globally over the entire region of interest in the state space, which are not straightforward to handle.
Sum-of-squares (SOS) optimization has been widely used to jointly search for barrier certificates and backup controllers, both in deterministic and stochastic domains \cite{steinhardt2012finite, prajna2007framework, santoyo2021barrier}. 
To address imperfect knowledge of the dynamics model, more recent approaches also employ other function approximation methods such as neural networks \cite{liu2023safe, mathiesen2023safety, tayal2024learning, zhang2025stochastic, neustroev2025neural} or Gaussian processes \cite{jagtap2020control, castaneda2021pointwise, khan2022gaussian} to represent barrier functions.
These works succeeded in finding or approximating valid SCBFS that provide certified lower bounds for the safety probability, but they do not necessarily imply that the SCBF values are ordered monotonically with the actual probabilities of long-horizon safety.
In particular, the long-horizon behavior of the gap between the lower bound and the safety probability has received relatively little attention.

\subsubsection{Hamilton-Jacobi Reachability Analysis}

Hamilton-Jacobi (HJ) reachability analysis offers a constructive way of computing the barrier certificates \cite{bansal2017hamilton, choi2021robust}.
This framework yields a function that serves as both a formal safety certificate and a meaningful measure of safety tied to a certain user-defined cost, typically chosen as the distance to the nearest failure state. This is done by formulating an optimal control problem with respect to the minimum or maximum cost value along the closed-loop trajectory.
Thanks to its close connection to optimal control, HJ reachability admits a wide range of formulations and extensions, e.g., data-driven approaches \cite{choi2025data}, reinforcement learning \cite{fisac2019bridging, ganai2024hamilton}.
However, these deterministic formulations do not directly characterize safety probabilities for systems subject to stochastic disturbances.
Stochastic extensions of HJ reachability methods, such as \cite{althoff2008stochastic, esfahani2016stochastic, assellaou2018hamilton}, address this by formulating the value function as a probability of reaching the set of unsafe states.
Still, when these safety probabilities asymptotically decay toward zero, which is common in stochastic systems, they lose their ability to meaningfully evaluate safety over an infinite time horizon.

\subsubsection{Operator-Theoretic Methods in Control}

Operator-theoretic methods, particularly those grounded in the Koopman operator theory \cite{brunton2022modern}, have emerged as a powerful tool for analyzing nonlinear systems by \textit{lifting} the dynamics to an infinite-dimensional space of scalar observables where their evolution becomes linear \cite{korda2018linear, siguenza2025control}. 
This enables the application of linear control techniques to complex nonlinear dynamics.
This capability has been made practically accessible through data-driven methods such as dynamic mode decomposition \cite{schmid2010dynamic} and its extensions \cite{williams2015data, proctor2016dynamic}.
While these methods have also been adapted for safety-critical control \cite{folkestad2021data, zinage2023neural}, they primarily employ Koopman operators for constructing coordinate transformations to linear predictive models on the lifted space, rather than studying the evolution of safety probabilities.

\subsection{Notation}

Let $\Omega$ be a nonempty, bounded, open, and connected subset of a Euclidean space, i.e., a bounded domain.
We denote by $B_b(\Omega)$ the space of bounded Borel-measurable real-valued functions on $\Omega$, and by $C_0(\Omega)$ the space of continuous functions vanishing at the boundary $\partial\Omega$.
Functions belonging to $C_0(\Omega)$ are also understood as functions on $\overline{\Omega}$ through their continuous extensions with zero boundary values.
The space $C^k(\Omega)$ consists of functions whose derivatives up to order $k$ are continuous in $\Omega$.
The Sobolev space $W_\mathrm{loc}^{k,p}(\Omega)$ consists of functions whose weak derivatives up to order $k$ belong to $L_\mathrm{loc}^p(\Omega)$, i.e., are locally $p$-integrable.
The function spaces $B_b(\Omega)$ and $C_0(\Omega)$ are understood as Banach spaces, equipped with the supremum norm $\norm{\beta} = \sup_{x \in \Omega} |\beta(x)|$.

For a linear operator $L$ defined on a Banach space, $D(L)$ denotes its domain, $L^*$ its adjoint, $\norm{L}$ the induced operator norm.
The spectrum of a linear operator $L$, denoted by $\boldsymbol{\sigma}(L)$, is the set of complex numbers $z$ for which $z\operatorname{Id} - L$ (where $\operatorname{Id}$ is the identity map) does not have a bounded inverse defined on the entire space.
For a bounded operator $L$, its spectral radius is $\boldsymbol{r}(L) \coloneqq \sup_{\lambda\in\boldsymbol{\sigma}(L)} |\lambda|$.
Spectral quantities for real operators are understood through their complexifications.

For a real-valued function $\beta$ on $\Omega$, $\beta \geq 0$ means $\beta(x) \geq 0$ for all $x \in \Omega$, and $\beta > 0$, $\beta = 0$ is defined likewise.
Note that $\beta \geq 0$ and $\beta \neq 0$ (meaning $\neg (\beta = 0)$) does not imply $\beta > 0$.
A linear operator $L$ is called positive if $L\beta \geq 0$ whenever $\beta \geq 0$, and positivity improving if $L\beta > 0$ for any $0 \neq \beta \geq 0$.

\section{Problem Formulation} \label{sec: problem formulation}

\subsection{Stochastic Safety-Critical Dynamics}

In this paper, we consider the stochastic dynamics written in the form of an It\^o stochastic differential equation (SDE):
\begin{equation} \label{eq: dynamics}
    dX_t = f(X_t, u_t) dt + \sigma(X_t, u_t) dW_t,
\end{equation}
where $X_t \in \mathbb{R}^{n_x}$, $u_t \in U$ are the system's state and input at time $t$, $U$ is the set of admissible inputs which is assumed to be a compact subset of $\mathbb{R}^{n_u}$, $f:\mathbb{R}^{n_x}\times \mathbb{R}^{n_u} \rightarrow \mathbb{R}^{n_x}$ and $\sigma:\mathbb{R}^{n_x} \times \mathbb{R}^{n_u} \rightarrow \mathbb{R}^{{n_x}\times n_w}$ are locally Lipschitz functions.
The stochastic process $\{X_t\}_{t\geq0}$ is defined on the filtered probability space $(\Xi, \mathcal{F}, \{\mathcal{F}_t\}_{t \geq 0}, \mathbb{P})$ satisfying the usual conditions, where $\Xi$ is the sample space, $\mathcal{F}$ denotes its $\sigma$-algebra, $\{\mathcal{F}_t\}_{t\geq 0}$ is the filtration, and $\mathbb{P}$ is the probability measure on $(\Xi, \mathcal{F})$.
Here, $W_t$ is an $n_w$-dimensional $\{\mathcal{F}_t\}_{t \geq 0}$-Brownian motion and $X_t$ is $\{\mathcal{F}_t\}_{t \geq 0}$-adapted.

Suppose the system is subject to a safety constraint $X_t \in \mathcal{C}$, where $\mathcal{C}$ is a bounded subset of $\mathbb{R}^{n_x}$.
It can be interpreted as the set of allowable states, for which we seek to upper-bound the probability of trajectory escape.
Let $\Omega \coloneqq \operatorname{Int} \mathcal{C}$. We make the following assumption on $\Omega$ and $\mathcal{C}$.

\begin{assumption} \label{assumption: safe set}
    The safe set $\mathcal{C}$ satisfies $\mathcal{C} = \overline{\Omega}$, and $\Omega$ is a bounded, connected domain satisfying the uniform exterior cone condition.
\end{assumption}

The uniform exterior cone condition allows a broad class of domains with possibly nonsmooth boundaries while providing the boundary regularity required for the subsequent elliptic operator analysis.
It covers most of the practical safe set models, allowing corners and nonconvexities, because it only fails under pathological cases such as an infinitely sharp inward spine.

When $X_t$ first reaches the boundary of $\mathcal{C}$, then all further behavior of the system is regarded as unsafe.
To represent this, we define the first exit time from $\Omega$ as
\begin{equation}
    \tau_\Omega \coloneqq \inf \{t > 0 : X_t \notin \Omega\},
\end{equation}
where $\inf \varnothing = \infty$. Note that $\tau_\Omega$ is a random variable which is correlated with the feedback law and the initial condition $X_0$.
Considering a potentially time-varying policy $\pi:[0,\infty) \times \mathcal{C} \rightarrow U$, the safety probability as a function of time and the initial condition under the feedback law $u_t = \pi(t, X_t)$ is defined as
\begin{equation}
    Z^\pi (t,x) \coloneqq \mathbb{P}^\pi[\tau_\Omega > t | X_0 = x] = \mathbb{E}^\pi [1_{\{\tau_\Omega > t\}}| X_0 = x],
\end{equation}
for $x \in \mathcal{C}$.
We use $\mathbb{P}^\pi$ and $\mathbb{E}^\pi$ to denote probability and expectation, respectively, under the closed-loop process induced by $\pi$ on $(\Xi,\mathcal{F})$.
Thus, $Z^\pi(t,x)$ is the probability that the trajectory initialized at $x$ remains in the safe domain $\Omega$ throughout the interval $[0,t]$. It is straightforward to find that $Z^\pi$ is nonnegative and nonincreasing with respect to $t$, and $Z^\pi(0,x) = 1_\Omega(x)$.
Hereafter, we will consider time-invariant policies $\pi:\mathcal{C}\rightarrow U$ that are Lipschitz continuous, unless stated otherwise.

Since the diffusion term $\sigma(X_t, u_t) dW_t$ introduces stochastic disturbances, sample trajectories starting at the same state and following the same policy may have different safety outcomes.
Consequently, safety under stochastic dynamics is naturally characterized probabilistically.
Unlike in the deterministic case where safe states (those contained in an invariant subset of $\mathcal{C}$) and unsafe states (those lying outside any such set) can be clearly distinguished, in stochastic systems, safety is generally not a binary property. 
Instead, the safety probability $Z^\pi(t,x)\in[0,1]$ provides a graded measure of safety for a given initial state and time horizon.

\subsection{Stochastic Control Barrier Functions}

The SCBF is a stochastic extension of the CBF formulation \cite{ames2019control}. Wherever $\beta:D(\beta)(\subseteq \mathbb{R}^{n_x})\rightarrow \mathbb{R}$ is twice differentiable, define the differential operator associated with the controlled It\^o diffusion by
\begin{multline} \label{eq: differential operator}
    \mathcal{A}^u \beta (x) 
    \coloneqq \partial_x \beta(x) \cdot f(x,u) \\
    {} + \frac{1}{2} \operatorname{Tr} \left(\partial_{xx} \beta(x) \sigma(x,u) \sigma(x,u)^\top \right).
\end{multline}
Given a feedback policy $\pi$, we likewise write
\begin{equation} \label{eq: L pi definition}
    \mathcal{A}^\pi \beta(x) \coloneqq \mathcal{A}^{\pi(x)} \beta(x), \quad \forall x.
\end{equation}
This differential operator plays the role in stochastic dynamics analogous to the directional derivative in deterministic systems. It\^o's lemma says that if $\beta$ is twice continuously differentiable in the neighborhood of $x$,
\begin{equation}
    \mathcal{A}^\pi \beta (x) = \lim_{t \searrow 0} \frac{\mathbb{E}^\pi[\beta(X_t)|X_0=x] - \beta(x)}{t}.
\end{equation}

SCBFs for It\^o diffusion processes are defined in many different ways across the literature, although they share the similar structure. In this paper, we use the following \textit{zeroing-type} definition for SCBFs.

\begin{definition}[Stochastic CBF] \label{def: scbf}
    A scalar function $h:\mathcal{C} \rightarrow \mathbb{R}$ is an SCBF if it is twice continuously differentiable in $\Omega$ and continuous on $\mathcal{C}$, $\max_{x \in \mathcal{C}} h(x) > 0$, and there exists a Lipschitz continuous backup policy $\pi_b:\mathcal{C}\rightarrow U$ such that
    \begin{align}
        h(x) \leq 0 & \quad \forall x \in \partial \mathcal{C}, \label{eq: scbf condition 1} \\
        \mathcal{A}^{\pi_b} h(x) + \gamma h(x) \geq 0 &\quad \forall x \in \Omega, \label{eq: scbf condition 2}
    \end{align}
    where $\gamma \geq 0$ is a pre-specified parameter which we call the \textit{SCBF decay rate}.
\end{definition}

While a deterministic CBF provides control invariance guarantee for its super zero level set through Nagumo's theorem \cite[Section 4.2]{blanchini2008set}, an SCBF provides a lower bound on the probability of the system not escaping from the set $\mathcal{C}$. Specifically, a Lipschitz continuous policy satisfying \autoref{eq: scbf condition 2} for an SCBF $h$ can constrain the system's state $X_t$ within $\mathcal{C}$ at a probability with a probability bounded below by an exponentially decaying function.

\begin{theorem}[Safety Probability Lower Bound] \label{thm: scbf prob bound}
    Suppose $h$ is an SCBF and $\pi : [0,\infty) \times \mathcal{C} \rightarrow U$ is a possibly time-varying control policy, Lipschitz continuous in state, satisfying
    \begin{equation} \label{eq: scbf condition time-varying}
        \mathcal{A}^{\pi(t,x)} h(x) + \gamma h(x) \geq 0 
    \end{equation}
    for all $t \geq 0$ and $x\in \Omega$. Then, for all $x\in \Omega$,
    \begin{equation}
        Z^\pi(t,x) \geq \frac{h(x)}{\max_{y\in \mathcal{C}} h(y)}\cdot e^{-\gamma t}.
    \end{equation}
\end{theorem}
\begin{proof}
It is an immediate corollary of \cite[Proposition 1]{santoyo2019verification} which is based on \cite[Theorem 3.1]{kushner1967stochastic}.
\end{proof}

Inspired by \cite{cbfqp}, a straightforward way of building a policy that satisfies the SCBF condition \autoref{eq: scbf condition time-varying} is to formulate the optimization-based controller
\begin{equation} \label{eq: scbf-qp}
\begin{aligned}
    \pi(t,x) = \arg \min_{u \in U} \quad & \norm{u - \pi_\mathrm{ref}(t,x)}^2 \\
    \mathrm{s.t.} \quad & \mathcal{A}^u h(x) + \gamma h(x) \geq 0,
\end{aligned}
\end{equation}
where $\norm{\cdot}$ is a weighted two-norm on $\mathbb{R}^{n_u}$ and $\pi_\mathrm{ref}$ is the reference input which is not necessarily safe.
If $h$ is a valid SCBF, then \autoref{eq: scbf-qp} is guaranteed to be feasible for $x \in \Omega$ because $u = \pi_b(x)$ is a feasible solution.
In addition, if $\sigma$ does not explicitly depend on the input and the drift part $f(\cdot, \cdot)$ is affine in input, then \autoref{eq: scbf-qp} becomes a QP, namely SCBF-QP, which can be solved at a relatively low computational cost.

Despite the guarantees above, constructing a valid SCBF is generally challenging.
The SCBF condition requires finding both a sufficiently regular function satisfying the boundary condition \autoref{eq: scbf condition 1} and a backup policy such that the differential inequality \autoref{eq: scbf condition 2} holds everywhere in the safe domain. 
Moreover, since the connection between the SCBF and the safety probability is given only through a one-sided inequality in \autoref{thm: scbf prob bound}, the SCBF value cannot in general serve as a consistent measure of safety.
A larger SCBF value provides a higher certified lower bound but does not necessarily imply a higher actual safety probability.
These limitations motivate a systematic construction of an SCBF from the closed-loop stochastic dynamics rather than relying on a prescribed candidate function.
Ideally, its values should be meaningfully tied to the evolution of $Z^\pi$, rather than satisfying only a sufficient inequality condition.

\section{The Operator-Theoretic Approach}

To construct an SCBF whose values are directly related to the safety probability, we begin by studying how $Z^\pi$ evolves under a fixed feedback policy.
We first show that the time evolution of $Z^\pi$ is described by a family of linear operators acting on a space of scalar functions of state, which forms a killed semigroup associated with the first exit from the safe domain $\Omega$.
We then connect this semigroup to the differential operator $\mathcal{A}^\pi$ through its infinitesimal generator.

\subsection{Dynamic Programming for Safety Probability}

Suppose a time-invariant policy $\pi : \mathcal{C} \rightarrow U$ is given as a Lipschitz continuous function of state. We first derive the dynamic programming principle governing the time evolution of $Z^\pi (t,x)$.
For any $t, s > 0$, the survival until time $t+s$ requires the trajectory first to survive until time $s$, and, conditional on this event, to remain in $\Omega$ for an additional time $t$. 
Since the closed-loop diffusion process under the time-invariant policy $\pi$ is a time-homogeneous Markov process, the probability of surviving for the additional duration depends on the past only through the state $X_s$. 
More precisely, by the Markov property,
\begin{equation}
    \mathbb{E}^\pi \left[1_{\{\tau_\Omega > t+s\}} \middle| \mathcal{F}_s \right] = 1_{\{\tau_\Omega > s\}} Z^\pi (t,X_s).
\end{equation}
Therefore, 
\begin{equation} \label{eq: dynamic programming}
    Z^\pi (t+s, x) = \mathbb{E}^\pi \left[Z^\pi (t,X_s) \cdot 1_{\{\tau_\Omega > s\}}\middle| X_0=x \right],
\end{equation}
which serves as the dynamic programming principle for $Z^\pi$.

Motivated by \autoref{eq: dynamic programming}, define a \textit{linear} operator $\mathcal{T}_t^\pi$ by
\begin{equation} \label{eq: T definition}
    \mathcal{T}_t^\pi \beta (x) \coloneqq \mathbb{E}^\pi \left[\beta(X_t) \cdot 1_{\{\tau_\Omega > t\}}\middle| X_0 = x\right],
\end{equation}
for any $\beta \in B_b(\Omega)$.
Then the dynamic programming principle \autoref{eq: dynamic programming} can be rewritten as
\begin{equation}
    Z^\pi (t+s, \cdot) = \mathcal{T}_s^\pi Z^\pi(t,\cdot)
\end{equation}
for all $t, s \geq 0$, and in particular,
\begin{equation}
    Z^\pi (t,\cdot) = \mathcal{T}_t^\pi 1_\Omega,
\end{equation}
where $1_\Omega$ is the indicator function for $\Omega$.

The operator $\mathcal{T}_t^\pi$ fully describes the dynamic programming principle for $Z^\pi$. By studying its properties, one can capture the characteristics of how the safety probability evolves.
It is noteworthy that $\mathcal{T}_t^\pi$ is a linear map between the space of functions, and thus it is typically an infinite-dimensional operator.
While the usual properties of finite-dimensional vectors do not naturally extend to infinite dimensions, the remainder of this paper heavily relies on operator theory and functional analysis.
We summarize some basic properties of $\mathcal{T}_t^\pi$ in the following theorem.

\begin{theorem}[Basic Properties of $\mathcal{T}_t^\pi$] \label{thm: T basic properties}
Consider a fixed policy $\pi : \mathcal{C} \rightarrow U$ Lipschitz continuous in state.
\begin{enumerate}
    \item $\mathcal{T}_t^\pi$ is a positive linear operator with $\norm{\mathcal{T}_t^\pi} \leq 1$.
    \item The collection $\mathcal{T}^\pi \coloneqq \{\mathcal{T}_t^\pi\}_{t \geq 0}$ forms a semigroup, i.e., $\mathcal{T}_0^\pi = \operatorname{Id}$ and $\mathcal{T}_{t+s}^\pi = \mathcal{T}_t^\pi \circ \mathcal{T}_s^\pi$ for all $t, s \geq 0$.
\end{enumerate}
\end{theorem}
\begin{proof}
    Linearity and positivity follow directly from the definition.
    Moreover, for every $\beta\in B_b(\Omega)$,
    \begin{equation}
        |\mathcal{T}_t^\pi\beta(x)|
        \leq
        \mathbb{E}^\pi        \left[|\beta(X_t)| \cdot 1_{\{\tau_\Omega>t\}} \middle| X_0 = x\right] \leq \norm{\beta},
    \end{equation}
    and hence $\norm{\mathcal{T}_t^\pi \beta} \leq \norm{\beta}$, i.e., $\norm{\mathcal{T}_t^\pi} \leq 1$.

    For the second claim, let $t, s > 0$ and $\beta \in B_b(\Omega)$. By the Markov property and time homogeneity of the closed-loop diffusion, 
    \begin{equation}
        \mathbb{E}^\pi \left[\beta(X_{t+s}) \cdot 1_{\{\tau_\Omega > t + s\}} \middle| \mathcal{F}_s\right] = 1_{\{\tau_\Omega > s\}} \mathcal{T}_t^\pi \beta (X_s).
    \end{equation}
    Taking the conditional expectation given $X_0=x$ and using the law of total expectation gives
    \begin{equation}
    \begin{aligned}
        \mathcal{T}_{t+s}^\pi \beta(x) &= \mathbb{E}^\pi \left[ \mathcal{T}_t^\pi \beta (X_s) \cdot 1_{\{\tau_\Omega > s\}} \middle| X_0 = x \right] \\
        &= \mathcal{T}_s^\pi (\mathcal{T}_t^\pi \beta) (x).
    \end{aligned}
    \end{equation}
    This should hold for every $\beta \in B_b(\Omega)$, so $\mathcal{T}_{t+s} = \mathcal{T}_t \circ \mathcal{T}_s$.
    Since $\mathcal{T}_0^\pi = \operatorname{Id}$ follows directly from the definition, we conclude that $\mathcal{T}^\pi = \{\mathcal{T}_t^\pi\}_{t \geq 0}$ forms a semigroup.
\end{proof}

Since $\mathcal{T}^\pi = \{\mathcal{T}_t^\pi\}_{t \geq 0}$ is a semigroup that is associated with the first exit time $\tau_\Omega$ from $\Omega$, it is called the killed semigroup.

\subsection{The Killed Semigroup and Its Infinitesimal Generator}

In this subsection, we establish the connection between the semigroup $\mathcal{T}^\pi$ and the differential operator $\mathcal{A}^\pi$ introduced in \autoref{eq: L pi definition}.
To this end, consider the restriction of $\mathcal{T}_t^\pi$
\begin{equation}
    \left. \mathcal{T}_t^\pi \right|_{C_0(\Omega)} : C_0(\Omega) \rightarrow B_b(\Omega),
\end{equation}
where $C_0(\Omega)$ denotes the Banach space, more precisely a Banach lattice, of continuous real functions on $\Omega$ vanishing on the boundary $\partial \Omega$, equipped with the supremum norm.
We will show below that $\mathcal{T}_t^\pi$ leaves $C_0(\Omega)$ invariant. Its restriction to $C_0(\Omega)$ will be denoted by
$T_t^\pi \coloneqq \mathcal{T}_t^\pi|_{C_0(\Omega)} : C_0(\Omega) \rightarrow C_0(\Omega)$.
For this, we first view the differential expression \autoref{eq: L pi definition} as a differential operator in the Sobolev sense\footnote{The choice of the $L^{n_x}$-norm here allows the results from \cite{arendt2014semigroups} to be applied in \autoref{thm: dirichlet realization}.}
\begin{equation}
    \mathcal{A}^\pi : W_\mathrm{loc}^{2, n_x}(\Omega) \rightarrow L_\mathrm{loc}^{n_x} (\Omega).
\end{equation}
We also define its Dirichlet realization $A_D^\pi$ on $C_0(\Omega)$ by
\begin{equation}
    A_D^\pi : D(A_D^\pi) \subset C_0(\Omega) \rightarrow C_0(\Omega),
\end{equation}
where
\begin{equation} \label{eq: generator domain definition}
    D(A_D^\pi) \coloneqq \{\beta \in C_0(\Omega) \cap W_\mathrm{loc}^{2, n_x}(\Omega) : \mathcal{A}^\pi \beta \in C_0(\Omega)\}
\end{equation}
and $A_D^\pi \beta = \mathcal{A}^\pi \beta$ almost everywhere in $\Omega$.

At this point, we introduce an important assumption:
\begin{assumption} \label{assumption: full rank diffusion}
    For every $(x,u) \in \mathcal{C} \times U$, the matrix $\sigma(x,u)$ has full row rank.
\end{assumption}
This defines the class of stochastic systems considered in this paper, where the diffusion coefficients associated with different state dimensions are linearly independent.
Such models can arise when uncertainty across all components of the modeled state is represented by independent diffusion, including in many stochastic reduced-order models of higher-dimensional real-world systems.
Combined with \autoref{assumption: safe set}, it implies that $\mathcal{A}^\pi$ is a uniformly elliptic operator on $\Omega$ for any admissible Lipschitz continuous $\pi$, allowing one to apply the well-established elliptic operator theory in the subsequent analysis.
Dynamic models with rank-deficient diffusion fall outside the formal guarantees that will be presented, although we will later empirically examine such cases through numerical examples.

\begin{theorem} \label{thm: dirichlet realization}
For every $t \geq 0$, $\mathcal{T}_t^\pi$ leaves $C_0(\Omega)$ invariant. Its restriction on $C_0(\Omega)$ 
\begin{equation}
    T_t^\pi : C_0(\Omega) \rightarrow C_0(\Omega)
\end{equation}
is therefore well-defined as a linear operator on $C_0(\Omega)$. 
The Dirichlet realization $A_D^\pi$ generates a bounded, holomorphic, positive, and irreducible $C_0$-semigroup\footnote{A $C_0$-semigroup denotes a strongly continuous semigroup of operators, distinct from the function space $C_0(\Omega)$.} $T^\pi \coloneqq \{T_t^\pi\}_{t \geq 0}$.
\end{theorem}
\begin{proof}
    Suppose $n_x \geq 2$. 
    Under the stated assumptions, the second-order coefficients of $\mathcal{A}^\pi$ are continuous on $\mathcal{C}$ and uniformly elliptic on $\Omega$, while the first-order coefficients are bounded.
    Therefore, \cite[Theorem 4.1]{arendt2014semigroups} applies, showing that $A_D^\pi$ generates a bounded holomorphic $C_0$-semigroup on $C_0(\Omega)$.
    Moreover, \cite[Proposition 4.7]{arendt2014semigroups} shows that this semigroup is positive and irreducible.

    For $n_x = 1$, $\Omega$ is a bounded open interval and $\mathcal{A}^\pi$ is a one-dimensional uniformly elliptic operator in divergence form with bounded coefficients. 
    Thus, \cite[Section 4.1, Theorem 4.1]{arendt2020strict} applies to its Dirichlet realization, yielding a holomorphic, positive, and irreducible $C_0$-semigroup on $C_0(\Omega)$.

    Let $S^\pi = \{S_t^\pi\}_{t \geq 0}$ denote the semigroup generated by $A_D^\pi$.
    By the Feynman-Kac representation for the parabolic Dirichlet problem \cite[Theorem 3.2]{zhang2001parabolic}, for any $\beta\in C_0(\Omega)$ and $t\geq 0$,
    \begin{equation}
        S_t^\pi\beta(x) = \mathbb{E}^\pi
        [\beta(X_t) 1_{\{\tau_\Omega > t\}} | X_0=x ] = \mathcal{T}_t^\pi \beta(x).
    \end{equation}
    Consequently, $\mathcal{T}_t^\pi$ leaves $C_0(\Omega)$ invariant, and its restriction $T_t^\pi$ coincides with $S_t^\pi$.
\end{proof}

\autoref{thm: dirichlet realization} connects the probabilistic and differential descriptions of the dynamics.
In other words, the expectation defining $T_t^\pi\beta$ is equivalent to solving the parabolic Dirichlet problem associated with $\mathcal{A}^\pi$.
This allows us to use tools from elliptic PDEs and positive semigroup theory to analyze the long-time behavior of the safety probability.
The $C_0(\Omega)$-restriction possesses two additional properties that are useful in the spectral analysis.

\begin{lemma} \label{thm: T properties advanced}
For any $t > 0$, $T_t^\pi$ is positivity improving and compact.
\end{lemma}
\begin{proof}
    By \autoref{thm: dirichlet realization},
    $T^\pi$ is a positive, irreducible, holomorphic $C_0$-semigroup on $C_0(\Omega)$.
    Therefore, \cite[Theorem 2.3]{arendt2020strict} implies that $T_t^\pi$ is positivity improving for every $t>0$.
    For compactness, we invoke \cite[Proposition 4.7]{arendt2014semigroups} that $A_D^\pi$ has compact resolvent if $n_x \geq 2$.
    For $n_x = 1$, the resolvent of $A_D^\pi$ is the solution operator of a uniformly elliptic second-order ODE on a bounded interval, subject to Dirichlet boundary conditions.
    It therefore admits an integral representation with a continuous Green's function, and is compact on $C_0(\Omega)$.
    Since a holomorphic $C_0$-semigroup is immediately norm-continuous, the compactness of its resolvent implies that $T_t^\pi$ is compact for every $t > 0$ \cite[Theorem II.4.29]{engel2000one}.
\end{proof}

These properties have direct consequences for the long-horizon safety analysis developed next.
The implication of positivity improvement is that surviving trajectories can reach any interior region (a nonempty open subset of $\Omega$) from any interior state with nonzero probability.
Compactness allows the evolution of the killed semigroup to be analyzed through isolated spectral modes.
Together, these properties play a central role in identifying a unique positive dominant eigenmode, which describes the asymptotic decay characteristics of the safety probability.

\subsection{Computation of $T_t^\pi$}

Since $A_D^\pi$ is the infinitesimal generator of the holomorphic $C_0$-semigroup $T^\pi$, for every $\beta \in C_0(\Omega)$ and $t > 0$, $T_t^\pi \beta \in D(A_D^\pi)$ and
\begin{equation}
    \frac{d}{dt} T_t^\pi \beta = A_D^\pi T_t^\pi \beta = T_t^\pi A_D^\pi \beta,
\end{equation}
where the time derivative is taken in the $C_0(\Omega)$ norm, and the well-definedness of the second equality is contingent upon $\beta \in D(A_D^\pi)$ \cite[Section II.4]{engel2000one}.
This provides a constructive way of evaluating $T_t^\pi$ through a parabolic PDE. For $\beta \in C_0(\Omega)$, let $b:[0,t]\times \overline{\Omega} \rightarrow \mathbb{R}$ be a continuous function which satisfies the initial-boundary value problem
\begin{equation} \label{eq: T pde}
\begin{aligned}
    \partial_\tau b(\tau, x) &= \mathcal{A}^\pi b(\tau, x), && (\tau, x) \in (0,t] \times \Omega \\
    b(0,x) &= \beta(x), && x \in \Omega \\
    b(\tau, x) &= 0, && (\tau, x) \in [0,t] \times \partial \Omega,
\end{aligned}
\end{equation}
solved forwards in time.
This PDE admits a unique solution $b(\tau, \cdot) = T_\tau^\pi \beta$ for all $\tau \in [0,t]$ \cite[Section II.4, Proposition II.6.2]{engel2000one}, giving $T_t^\pi \beta = b(t,\cdot)$.

We also provide the following lemma which shows that the PDE \autoref{eq: T pde} has a spatial smoothing property.
Even if the initial condition is only continuous, the solution becomes twice continuously differentiable everywhere in $\Omega$ for any positive time.
This allows the operator $\mathcal{A}^\pi$ to be evaluated pointwise in the classical sense.
It also enables the spatial derivatives to be approximated using finite difference methods, with truncation errors vanishing as the grid spacing decreases.

\begin{lemma} \label{lemma: semigroup regularity}
    For any $\beta \in C_0(\Omega)$ and $t>0$,
    $T_t^\pi\beta \in C^2(\Omega)
    \cap C_0(\Omega)$, provided that $\pi$
    is Lipschitz continuous.
\end{lemma}
\begin{proof}
    As established above, $\zeta \coloneqq T_t^\pi \beta \in D(A_D^\pi)$. This directly gives $\zeta \in C_0(\Omega)$.
    Moreover, $g \coloneqq A_D^\pi \zeta$ also belongs to $D(A_D^\pi)$ because $g = T_{t/2}^\pi A_D^\pi T_{t/2}^\pi \beta$ and
    \begin{equation}
    \begin{aligned}
        T_{t/2}^\pi \beta \in D(A_D^\pi) \;&\Longrightarrow \; A_D^\pi T_{t/2}^\pi \beta \in C_0(\Omega) \\
        &\Longrightarrow \; g = T_{t/2}^\pi A_D^\pi T_{t/2}^\pi \beta \in D(A_D^\pi).
    \end{aligned}
    \end{equation}
    This means $\mathcal{A}^\pi \zeta = g$ almost everywhere in $\Omega$.

    The coefficients of $\mathcal{A}^\pi$, namely $\sigma(x, \pi(x)) \sigma(x, \pi(x))^\top$ and $f(x, \pi(x))$, are locally Lipschitz continuous under the stated assumptions.
    For $n_x\geq 2$, $g\in D(A_D^\pi) \subset W_\mathrm{loc}^{2, n_x} \subset  W_\mathrm{loc}^{1,q}(\Omega)$ for any finite $q > 1$.
    Thus, \cite[Theorem 9.19]{gilbarg1998elliptic} yields $\zeta \in W_\mathrm{loc}^{3,q}(\Omega)$.
    We can choose $q > n_x$, then the Sobolev embedding \cite[Theorem 7.26]{gilbarg1998elliptic} gives $\zeta \in C^2(\Omega)$.
    The $n_x = 1$ case cannot be covered via the same logic.
    Still, $\zeta \in W_\mathrm{loc}^{2, n_x}(\Omega)$ implies that $\zeta'$ is continuous.
    Since $\mathcal{A}^\pi \zeta = g$ and $\sigma(x,\pi(x))\sigma(x,\pi(x))^\top > 0$ for all $x \in \Omega$, we directly obtain
    \begin{equation}
        \zeta'' (x) = \frac{g(x) - \zeta'(x) \cdot f(x, \pi(x))}{\frac{1}{2} \sigma(x, \pi(x))\sigma(x,\pi(x))^\top}
    \end{equation}
    from \autoref{eq: differential operator}, which implies $\zeta''$ is continuous in $\Omega$.
\end{proof}

\section{The Dominant Eigenfunction as SCBF}

Recall that SCBFs are intended to provide guarantees on the safety probability $Z^\pi$.
For any fixed $t > 0$ and $n \in \mathbb{N}$, the semigroup property gives 
\begin{equation}
    Z^\pi(nt, x) = \mathcal{T}_{nt}^\pi 1_\Omega(x) = (\mathcal{T}_t^\pi)^n 1_\Omega(x),
\end{equation}
which is the same linear operator applied $n$ times to the initial condition $1_\Omega$.
This motivates asking whether there exists a unique dominant spectral mode of $\mathcal{T}_t^\pi$ that governs the repeated application of the operator as $n$ grows, and whether this gives an asymptotic approximation in the form of
\begin{equation} \label{eq: Z approximation conjecture n}
    Z^\pi (nt, x) = (\mathcal{T}_{t}^\pi)^n 1_\Omega(x) \stackrel{?}{=} c\cdot \lambda^n \phi (x) + o(\lambda^n),
\end{equation}
where $\lambda$ and $\phi$ are the (putative) dominant eigenvalue and eigenfunction of $\mathcal{T}_t^\pi$, respectively, $c$ is a constant independent of $n$, and $o(\lambda^n)$ represents the terms that decay asymptotically faster than the dominant mode as $n \rightarrow \infty$.

In this section, we make the proposed intuition rigorous by analyzing the restriction $T_t^\pi = \mathcal{T}_t^\pi|_{C_0(\Omega)}$, whose positivity improvement and compactness properties were established above.
These properties allow us to identify the dominant spectral mode of the infinite-dimensional operator $T_t^\pi$ and to show that it is unique, positive and independent of the chosen time horizon $t$.
We then show that this dominant mode characterizes the asymptotic behavior of the safety probability, whose time evolution is described by the semigroup $\mathcal{T}^\pi$.
Finally, we prove that the resulting dominant eigenfunction is a valid SCBF.

\subsection{Spectral Analysis}

We begin by fixing a time horizon $t> 0$ and studying the spectrum of $T_t^\pi$.
Above, \autoref{thm: T properties advanced} established two properties that are particularly important here: $T_t^\pi$ is compact and positivity improving.
Together, these properties allow us to establish the existence of a unique, strictly positive dominant eigenmode, separated from the rest of the spectrum.

\begin{theorem}[Dominant Eigenpair] \label{thm: dominant eigenpair}
    For every Lipschitz continuous policy $\pi:\mathcal{C}\rightarrow U$, the spectral radius of $T_t^\pi$
    \begin{equation}
        r_t^\pi \coloneqq \boldsymbol{r}(T_t^\pi)
    \end{equation}
    is a strictly positive, isolated, and algebraically simple eigenvalue of $T_t^\pi$ for all $t > 0$. Its corresponding $r_t^\pi$-eigenfunction $\psi_t^\pi \in C_0(\Omega) \setminus \{0\}$ can be chosen to satisfy $\psi_t^\pi(x) > 0$ for all $x \in \Omega$.
    Moreover, there exists a dual eigenfunctional $0 \neq \varphi_t^\pi \in C_0(\Omega)^*$ such that $(T_t^\pi)^* \varphi_t^\pi = r_t^\pi \varphi_t^\pi$, which can be chosen to satisfy $\varphi_t^\pi(g) > 0$ for all $0 \neq g \in C_0(\Omega)$, $g \geq 0$.
    Every eigenvalue of $T_t^\pi$ other than $r_t^\pi$ is strictly smaller than $r_t^\pi$ in modulus.
\end{theorem}
\begin{proof}
    Fix $t > 0$.
    In this proof, we will write $T = T_t^\pi$ and $r = \boldsymbol{r}(T_t^\pi)$ for notational simplicity. By \autoref{thm: T properties advanced}, $T$ is compact and positivity improving.

    First, we show $r > 0$. Pick a compact set $\omega \subset \Omega$ and a function $\beta \in C_0(\Omega)$ such that $\beta(x) > 0$ for all $x \in \operatorname{Int} \omega$ and $\beta(x) = 0$ elsewhere. Since $\omega$ is a compact set and $T\beta(x) > 0$ for all $x \in \omega$ due to positivity improvement, there exists a $c > 0$ such that $\beta(x) / T\beta (x) \leq c$ for all $x \in \omega$ and therefore $T\beta \geq c^{-1}\beta$. Thus, we get $r \geq c^{-1} > 0$ after applying Gelfand's formula: $\boldsymbol{r}(T) = \lim_{n \rightarrow \infty} \norm{T^n}^{\frac{1}{n}}$.

    Since $T$ is a compact and positive operator having strictly positive spectral radius $r$, we can apply the Krein-Rutman theorem \cite[Theorem 6.2]{krein_rutman} to get
    \begin{equation}
        T\psi = r\psi, \qquad T^*\varphi = r\varphi
    \end{equation}
    for some $0\neq \psi \geq 0$ and positive functional $0 \neq \varphi \in C_0(\Omega)^*$. 
    Any other eigenvalue of bigger modulus would contradict that $r$ is the spectral radius.
    Positivity improvement yields $\psi(x) = T\psi(x) / r > 0$ for all $x \in \Omega$ and $\varphi (g) = \frac{1}{r} (T^* \varphi) (g) = \frac{1}{r} \varphi (Tg) > 0$ for all $0 \neq g \geq 0$.
    Thus, $r$ is a dominant eigenvalue of $T$ and one of its eigenfunctions $\psi$ can be chosen such that $\psi(x) > 0$ for all $x \in \Omega$. Moreover, since $T$ is compact and $r > 0$, $r$ is isolated.

    Next, we will show that $r$ is the only eigenvalue having modulus $r$.
    Suppose $\lambda \in \mathbb{C}$ is a potentially complex eigenvalue of $T$ having modulus $r$, $0 \neq \phi \in C_0(\Omega, \mathbb{C}) = C_0(\Omega) + \mathrm{i}\cdot C_0(\Omega)$ is its corresponding eigenfunction.
    Since $|\lambda| = r$, $T|\phi| \geq |T\phi| = r|\phi|$, where the inequality follows directly from the definition of $T$ \autoref{eq: T definition}. We can apply the dual eigenfunctional $\varphi$ on $T|\phi| - r|\phi|$ to get
    \begin{equation}
        \varphi(T|\phi| - r|\phi|) = \varphi(T|\phi|) - r\varphi (|\phi|) = 0,
    \end{equation}
    i.e., $T|\phi| = |T\phi|$. Pick a $x_0 \in \Omega$ and let $z \coloneqq T\phi(x_0) / |T\phi(x_0)|$ where the denominator is guaranteed to be positive since $|T\phi(x_0)| = T|\phi|(x_0) > 0$ due to positivity improvement.
    Define $\xi(x) \coloneqq |\phi(x)| - \operatorname{Re}(\overline{z}\phi(x))$. Since $|z| = 1$, one can easily find $\xi(x) \geq 0$ for all $x \in \Omega$, with $T\xi (x_0) = 0$.
    If $\xi$ were nonzero somewhere in $\Omega$, positivity improvement would imply $T\xi(x) > 0$ for all $x \in \Omega$, but this is a contradiction for $x = x_0$. Thus, $\xi = 0$, implying $\phi = z|\phi|$, and $\lambda = r$.
    The same argument shows that any real $r$-eigenfunction of $T$ should have a constant sign everywhere in $\Omega$.

    Finally, we will show that $r$ is algebraically simple.
    We begin with geometric simplicity. Let $\psi, \phi \in C_0(\Omega)$ be two eigenfunctions for eigenvalue $r$. 
    Pick $x_1, x_2 \in \Omega$ and let $\zeta$ be defined as
    \begin{equation}
        \zeta(x) \coloneqq (\phi(x_1) + \phi(x_2)) \psi(x) - (\psi(x_1) + \psi(x_2)) \phi(x),
    \end{equation}
    so that $\zeta(x_1) = -\zeta(x_2)$.
    However, at the same time, $\zeta$ is an $r$-eigenfunction and should have the same sign everywhere on $C$, so $\zeta(x_1) = \zeta(x_2) = 0$, or equivalently, $\phi(x_2) \psi(x_1) = \phi(x_1) \psi(x_2)$.
    Since this should hold for any pair $x_1, x_2 \in \Omega$, $\phi$ should be a scalar multiplication of $\psi$, making $r$ geometrically simple. 
    Now, suppose $g$ is a rank-$2$ generalized $r$-eigenfunction for $T$. Since $r$ is geometrically simple, $(T-r \cdot \operatorname{Id}) g = a\psi$ for some $a \neq 0$. However, applying $\varphi$ on both sides gives $0=a \varphi(\psi)$, so such $g$ cannot exist. Since any generalized $r$-eigenfunction of rank greater than $2$ would yield one of rank $2$, $r$ is algebraically simple.
\end{proof}

Now that we have shown that $\psi_t^\pi$ should be positive real, we will hereafter assume that the dominant primal/dual eigenfunction(al)s are properly normalized as
\begin{equation} \label{eq: eigenfunction normalization}
    \norm{\psi_t^\pi} = 1, \quad \psi_t^\pi \geq 0, \quad \varphi_t^\pi (\psi_t^\pi) = 1.
\end{equation}

\autoref{thm: dominant eigenpair} guarantees that $r_t^\pi$ is the unique spectral value and eigenvalue of $T_t^\pi$ having maximal modulus.
Because $T_t^\pi$ is compact, this dominance can be strengthened to a uniformly positive spectral gap between $r_t^\pi$ and the remaining spectrum.

\begin{corollary}[Spectral Gap] \label{cor: spectral gap}
For every $t>0$, there exists $\epsilon > 0$ such that $|\lambda| \leq r_t^\pi - \epsilon$ for all $\lambda \in \boldsymbol{\sigma}(T_t^\pi) \setminus \{r_t^\pi\}$.
\end{corollary}
\begin{proof}
    By \autoref{thm: dominant eigenpair}, $r_t^\pi>0$ is isolated and the unique spectral value of $T_t^\pi$ having modulus $r_t^\pi$.
    Since $T_t^\pi$ is compact for every $t>0$ (\autoref{thm: T properties advanced}), all nonzero spectral values are isolated eigenvalues of finite algebraic multiplicity, and they can accumulate only at zero.
    Hence $\boldsymbol{\sigma}(T_t^\pi) \setminus \{r_t^\pi\}$ is compact, and letting $r_t^\pi - \epsilon = \max_{s \in \boldsymbol{\sigma}(T_t^\pi) \setminus \{r_t^\pi\}} |s|$ proves the claim.
\end{proof}

Together with \autoref{thm: dominant eigenpair}, the spectral gap shows that repeated applications of $T_t^\pi$ are asymptotically dominated by its unique dominant mode.
However, as safety probability is defined over continuous time horizons, it remains to determine whether the dominant eigenspace depends on the particular choice of $t$.
Below, we show that the continuous semigroup structure ensures that the dominant eigenspace remains invariant for arbitrary $t > 0$.

\begin{theorem} \label{thm: r semigroup}
    The dominant eigenfunction $\psi_s^\pi$ for $T_s^\pi$ ($s > 0$) is also a dominant eigenfunction for $T_t^\pi$ for any $t \geq 0$. Moreover, there exists a $\gamma^\pi \geq 0$, independent of $t$, such that $r_t^\pi = e^{-\gamma^\pi t}$ for all $t \geq 0$.
\end{theorem}
\begin{proof}
    For $t = 0$, $T_0^\pi = \operatorname{Id}$ and the claims are immediate.
    For any $s,t > 0$, the semigroup property gives
    \begin{equation}
        T_s^\pi (T_t^\pi \psi_s^\pi) = T_t^\pi (T_s^\pi \psi_s^\pi) = r_s^\pi \cdot (T_t^\pi \psi_s^\pi).
    \end{equation}
    Since the eigenspace of $T_s^\pi$ is one-dimensional and $\psi_s^\pi$ and $\psi_t^\pi$ are both nonzero positive (\autoref{thm: dominant eigenpair}), $T_t^\pi \psi_s^\pi = c_t \psi_s^\pi$ for some $c_t > 0$.
    We can apply $\varphi_t^\pi$ on both sides to get
    \begin{equation}
        \varphi_t^\pi (T_t \psi_s^\pi) = r_t^\pi \varphi_t^\pi (\psi_s^\pi) = c_t \varphi_t^\pi (\psi_s^\pi),
    \end{equation}
    and since $\varphi_t^\pi(\psi_s^\pi) \neq 0$, $c_t = r_t^\pi$. This shows that the dominant eigenvalues inherit the semigroup property extends as:
    \begin{equation}
        r_{t+s}^\pi \psi_s^\pi = T_{t+s}^\pi \psi_s^\pi = T_t^\pi T_s^\pi \psi_s^\pi = r_t^\pi r_s^\pi \psi_s^\pi,
    \end{equation}
    and hence
    \begin{equation} \label{eq: semigroup r}
        r_{t+s}^\pi = r_t^\pi r_s^\pi.
    \end{equation}
    We can also see from \autoref{eq: semigroup r} that $a(t) \coloneqq -\log r_t^\pi$ (which is well-defined for $t \geq 0$ because $r_t^\pi > 0$) is a nondecreasing function of $t$ satisfying $a(s+t) = a(s) + a(t)$. Thus, $a(t) = \gamma^\pi t$ for a $\gamma^\pi \geq 0$, and $r_t^\pi = e^{-\gamma^\pi t}$. 
\end{proof}

Since the dominant eigenfunction is independent of the horizon $t$, hereafter, we use the notation $\psi^\pi$ to denote the normalized common dominant eigenfunction of $T_t^\pi$ for \textit{any} $t > 0$.
Moreover, since $A_D^\pi$ is the infinitesimal generator of $T^\pi$ (\autoref{thm: dirichlet realization}), we easily find
\begin{equation}
    A_D^\pi \psi^\pi = \lim_{t \searrow 0} \frac{T_t^\pi \psi^\pi - \psi^\pi}{t} = -\gamma^\pi \psi^\pi,
\end{equation}
that is, $\psi^\pi \in D(A_D^\pi)$ and
\begin{equation} \label{eq: scbf sobolev}
    \mathcal{A}^\pi \psi^\pi (x) = -\gamma^\pi \psi^\pi (x)
\end{equation}
almost everywhere in $\Omega$.

\subsection{Asymptotic Interpretation of the Dominant Eigenpair} \label{sec: asymptotic interpretation}

Now, we use dominant eigenpair to characterize the asymptotic behavior of the safety probability.
One technical point here is that $1_\Omega$ does not belong to $C_0(\Omega)$, whereas the spectral results above were established for $T_t^\pi$, the restriction of $\mathcal{T}_t^\pi$ onto $C_0(\Omega)$.
However, we show that $\mathcal{T}_t^\pi 1_\Omega \in C_0(\Omega)$ for any $t > 0$.

\begin{lemma} \label{lemma: t_0}
    For every $t > 0$, $Z^\pi(t, \cdot) = \mathcal{T}_t^\pi 1_\Omega \in C_0(\Omega)$.
\end{lemma}
\begin{proof}
    Take $\beta_n \in C_0(\Omega)$ such that $\beta_n(x)$ is a monotone increasing sequence indexed by $n \in \mathbb{N}$ with $0 \leq \beta_n(x) \nearrow 1$ for every $x \in \Omega$.
    Since $T_t^\pi : C_0(\Omega) \rightarrow C_0(\Omega)$ is compact, $\{T_t^\pi \beta_n\} = \{\mathcal{T}_t^\pi \beta_n\}$ has a subsequence converging in $C_0(\Omega)$, say $T_t^\pi \beta_{n_k} \rightarrow g \in C_0(\Omega)$ as $k \rightarrow \infty$.
    On the other hand, by monotone convergence, 
    \begin{multline}
        T_t^\pi \beta_n(x) = \mathbb{E}^\pi [\beta_n(X_t) 1_{\{\tau_\Omega > t\}}|X_0 = x] \\
        {} \longrightarrow \mathbb{E}^\pi [1_{\{\tau_\Omega > t\}}|X_0 = x] = \mathcal{T}_t^\pi 1_\Omega(x).
    \end{multline}
    Thus, the $C_0(\Omega)$-limit $g$ must satisfy $g(x) = \mathcal{T}_t^\pi 1_\Omega(x)$ for all $x \in \Omega$. Thus, $\mathcal{T}_t^\pi 1_\Omega = g \in C_0(\Omega)$.
\end{proof}

\autoref{lemma: t_0} indicates that, although $Z^\pi(0,\cdot) = 1_\Omega$ has discontinuities on the boundary $\partial \Omega$ and does not belong to $C_0(\Omega)$, applying the semigroup $\mathcal{T}^\pi$ immediately smoothens it so that $Z^\pi(t,\cdot) \in C_0(\Omega)$ is continuous everywhere on $\mathcal{C} = \overline{\Omega}$ for arbitrary $t>0$. 
This allows us to pick an arbitrary $t_0 > 0$ and apply the spectral results from above (established for $T_t^\pi$) to $Z^\pi(t_0,\cdot)$ without having to leave $C_0(\Omega)$.

\begin{theorem} \label{thm: prob approximation}
    Let $\psi^\pi$ and $\gamma^\pi$ be the normalized common dominant eigenfunction and the dominant decay rate from \autoref{thm: r semigroup}.
    Then,
    \begin{equation} \label{eq: prob approximation}
        Z^\pi (t,\cdot) = c^\pi e^{-\gamma^\pi t} \psi^\pi + O(e^{- at})
    \end{equation}
    for some $a > \gamma^\pi$ in the $C_0(\Omega)$ norm as $t \rightarrow \infty$, with
    \begin{equation} \label{eq: c pi}
        c^\pi = e^{\gamma^\pi t_0} \varphi_{\Delta}^\pi (Z^\pi (t_0, \cdot)) \geq 1
    \end{equation}
    for any $t_0, \Delta > 0$.
\end{theorem}
\begin{proof}
    For notational simplicity, in this proof, we will write $g = \mathcal{T}_{t_0}^\pi 1_\Omega \in C_0(\Omega)$ (see \autoref{lemma: t_0}), $T = T_\Delta^\pi$, $r = r_\Delta^\pi$, $\psi = \psi^\pi$, $\varphi = \varphi_\Delta^\pi$, where $\psi$ and $\varphi$ are properly normalized according to \autoref{eq: eigenfunction normalization}.
    Let $P:C_0(\Omega) \rightarrow C_0(\Omega)$ be the rank-one projection operator onto $\operatorname{span}(\psi)$ defined as
    \begin{equation}
        P \beta = \varphi(\beta) \psi, \quad \forall \beta \in C_0(\Omega).
    \end{equation}
    Then one can find that $P^2 = P$, and $TP = PT = rP$, and $NP = PN = 0$ where $N \coloneqq T - rP$. Also, $P$ is compact because it is a bounded finite-rank operator, and so is $N$ which is a finite linear combination of two compact operators $P$ and $T$, so every nonzero element of $\boldsymbol{\sigma}(N)$ is an eigenvalue of $N$ and every nonzero eigenvalue of $N$ belongs to its spectrum.

    Suppose $Nv = \lambda v$ for some $\lambda \in \mathbb{C} \setminus \{0\}$ and $v \in C_0(\Omega, \mathbb{C}) \setminus \{0\}$. Then, $0 = PNv = P \cdot (\lambda v)$ so $Pv = 0$, and thus $Tv = (rP+N)v = \lambda v$, making $\lambda$ an eigenvalue of $T$.
    On the other hand, if $Tv = \lambda v$ for nonzero $\lambda$ and $v$, $\varphi(Tv) = (T^* \varphi)(v) = r\varphi (v) = \lambda \varphi (v)$, and thus $r = \lambda$ or $\varphi (v) = 0$. If $\lambda = r$, then $v$ should be a scalar multiple of $\psi$, and if $\lambda \neq r$, then $Pv = \varphi(v) \psi = 0$, thus $Nv = (T - rP)v = \lambda v$. Moreover, since $N\psi = (T-rP)\psi = 0$, $0 \in \boldsymbol{\sigma}(N)$. Therefore, $\boldsymbol{\sigma}(N) = \{0\} \cup (\boldsymbol{\sigma}(T) \setminus \{r\})$.

    Now, for $t \geq t_0 + \Delta$, let $t = t_0 + n\Delta + s$ for $n \in \{1,2,\cdots\}$ and $s \in \left[0,\Delta\right)$. Since $P^2 = P$ and $NP = PN = 0$,
    \begin{equation}
        T^n = (rP + N)^n = r^n P + N^n,
    \end{equation}
    which yields
    \begin{equation}
    \begin{aligned}
        Z^\pi (t,\cdot) = T_s^\pi T^n g &= r^n \varphi(g)\cdot T_s^\pi \psi + T_s^\pi N^n g \\
        &= r^n r_s^\pi \varphi(g) \psi + T_s^\pi N^n g \\
        &= e^{-\gamma^\pi t} \cdot e^{\gamma^\pi t_0}\varphi(g) \psi + T_s^\pi N^n g.
    \end{aligned}
    \end{equation}
    Thus, with $c^\pi = e^{\gamma^\pi t_0}\varphi(g)$,
    \begin{equation}
        Z^\pi(t,\cdot) = c^\pi e^{-\gamma^\pi t} \psi + T_s^\pi N^n g.
    \end{equation}
    Since $\psi$ is normalized, $0 \leq \psi \leq 1_\Omega$, we can see that $g = \mathcal{T}_{t_0}^\pi 1_\Omega \geq \mathcal{T}_{t_0}^\pi \psi = T_{t_0}^\pi \psi = e^{-\gamma^\pi t_0} \psi$.
    With normalization of $\varphi$, this gives $c^\pi \geq 1$.

    Finally, we will show the residual term decays asymptotically faster than $e^{-\gamma^\pi t}$.
    Since $\boldsymbol{\sigma}(N) = \{0\} \cup (\boldsymbol{\sigma}(T) \setminus \{r\})$, \autoref{cor: spectral gap} shows that there exists a $\rho < r$ such that $\boldsymbol{r}(N) < \rho < r = e^{-\gamma^\pi \Delta}$. By Gelfand's formula, there exists a positive $M_\rho$ such that $\norm{N^n} \leq M_\rho \rho^n$. Thus,
    \begin{equation} \label{eq: residual}
        \norm{T_s^\pi N^n g} \leq \norm{T_s^\pi} \cdot \norm{N^n} \cdot \norm{g} \leq M_\rho \rho^n
    \end{equation}
    since $\norm{T_s^\pi} \leq 1$ (\autoref{thm: T basic properties}), and $\norm{g} \leq 1$ (because $g(x)$ is a probability).
    Since $\rho^n / e^{-\gamma^\pi t} = e^{\gamma^\pi (t_0 + s)}(\rho / r)^n \rightarrow 0$ as $t \rightarrow \infty$ with $t = t_0 + n\Delta + s$, the residual term \autoref{eq: residual} is $o(e^{-\gamma^\pi t})$. Letting $a = -\log\rho / \Delta > \gamma^\pi$ completes the proof.
\end{proof}

Despite \autoref{eq: c pi}, $c^\pi$ is independent of $t_0$ and $\Delta$.
Suppose $(t_0, \Delta)$ and $(\tilde{t}_0, \tilde{\Delta})$ produce $c^\pi$ and $\tilde{c}^\pi$, respectively.
They should asymptotically approximate the same probability $Z^\pi$, not depending on the two parameters, by
\begin{equation}
    Z^\pi (t,\cdot) = c^\pi e^{-\gamma^\pi t} \psi^\pi + o(e^{-\gamma^\pi t}) = \tilde{c}^\pi e^{-\gamma^\pi t} \psi^\pi + o(e^{-\gamma^\pi t}),
\end{equation}
so we have $(c^\pi - \tilde{c}^\pi) \psi^\pi = o(1)$ as $t \rightarrow \infty$. Because the left hand side stays constant regardless of $t$ and $\psi^\pi \neq 0$, $c^\pi = \tilde{c}^\pi$.

\autoref{thm: prob approximation} provides two complementary interpretations of $\gamma^\pi$ and $\psi^\pi$.
First, for any $x \in \Omega$,
\begin{equation} \label{eq: asymptotic log Z}
    \lim_{t \rightarrow \infty} -\frac{\log Z^\pi (t,x)}{t} = \gamma^\pi.
\end{equation}
Thus, $\gamma^\pi$ characterizes the asymptotic exponential decay rate of the safety probability, which is common to all initial states in $\Omega$.
Also, \autoref{thm: scbf prob bound} implies that $\gamma^\pi$ serves as the minimum achievable SCBF decay rate for the given backup policy $\pi$: for such an SCBF $h$ having decay rate $\gamma$,
\begin{equation}
    \gamma^\pi = \lim_{t \rightarrow \infty} -\frac{\log Z^\pi(t,x)}{t} \leq \lim_{t \rightarrow \infty} -\frac{\log \frac{h(x)}{\max_{y \in \mathcal{C}} h(y)} e^{-\gamma t}}{t} = \gamma,
\end{equation}
wherever $h(x) > 0$.
Second, it is a direct consequence of \autoref{thm: prob approximation} that
\begin{equation} \label{eq: relative safety measure}
    \lim_{t \rightarrow \infty} \frac{Z^\pi (t,x_1)}{Z^\pi (t, x_2)} = \frac{\psi^\pi (x_1)}{\psi^\pi (x_2)}, \quad \forall x_1, x_2 \in \Omega.
\end{equation}
We can view the dominant eigenfunction $\psi^\pi$ as a measure that asymptotically characterizes the relative safety of different initial states.
In particular, under the normalization in \autoref{eq: eigenfunction normalization}, the value of $\psi^\pi(x)$ quantitatively represents the asymptotic safety probability at $x$ relative to that at a state where $\psi^\pi$ attains its maximum.
It can also be easily seen that $\psi^\pi$ is the only nonnegative scalar function, up to positive scaling, which provides such an interpretation.
Still, it should be noted that these interpretations are asymptotic and do not imply that the dominant eigenfunction accurately represents the safety probability over arbitrary finite time horizons, where contributions from other spectral modes may remain significant.

\subsection{The Dominant Eigenfunction as an SCBF}

The observation \autoref{eq: scbf sobolev} suggests that $\psi^\pi$ might qualify as a valid SCBF for the system.
Through the following theorem, we verify each requirement to  identify $\psi^\pi$ as a valid SCBF for the given stochastic system and the safety constraint.

\begin{theorem}[Dominant Eigenfunction is an SCBF] \label{thm: dominant eigenfunction is an scbf}
    For a Lipschitz continuous policy $\pi : \mathcal{C} \rightarrow U$, the common dominant eigenfunction $\psi^\pi$ of the semigroup $T^\pi$ is a valid SCBF with any decay rate $\gamma \geq \gamma^\pi$.
\end{theorem}
\begin{proof}
    It follows directly from applying \autoref{lemma: semigroup regularity} to $\beta = e^{\gamma^\pi t} \psi^\pi$ that $\psi^\pi \in C_0(\Omega) \cap C^2(\Omega)$, i.e., $\psi^\pi$ is twice continuously differentiable everywhere in $\Omega$.
    Since both sides of the almost-everywhere equality \autoref{eq: scbf sobolev} should be continuous on $\Omega$, it should hold everywhere.

    Now, we verify the inequality conditions.
    Since $\psi^\pi \in D(A_D^\pi) \subset C_0(\Omega)$, \autoref{eq: scbf condition 1} is satisfied with equality.
    Next, since $\psi^\pi$ is nonnegative in $\Omega$, for any $\gamma \geq \gamma^\pi$, $\mathcal{A}^\pi \psi^\pi (x) + \gamma \psi^\pi (x) = (\gamma - \gamma^\pi) \psi^\pi (x) \geq 0$, for all $x \in \Omega$, satisfying \autoref{eq: scbf condition 2} with $\pi_b = \pi$.
\end{proof}

\autoref{thm: dominant eigenfunction is an scbf} completes the connection between the dominant eigenfunction and the SCBF framework.
While $\psi^\pi$ was initially derived from the spectral analysis of the killed semigroup, it is also shown to satisfy the conditions required of an SCBF.
This addresses the motivating question posed in the introduction. Beyond just satisfying the SCBF conditions, $\psi^\pi$ conveys the quantitative asymptotic interpretation that we observed in \autoref{sec: asymptotic interpretation}.
Through \autoref{thm: scbf prob bound}, properly normalized $\psi^\pi$ (as an SCBF with decay rate $\gamma \geq \gamma^\pi$) provides the exponentially-decaying safety probability lower bound $Z^{\pi'} (t,x) \geq \psi^\pi(x)e^{-\gamma t}$ for any $\pi'$ satisfying $\mathcal{A}^{\pi'} \psi^\pi \geq -\gamma \psi^\pi$.

\section{Power-Policy Iteration} \label{sec: power-policy iteration}

Having seen that the common dominant eigenfunction of the semigroup $T^\pi$ is a valid SCBF, the next question is how to compute it. 
In this section, we introduce a power iteration method to compute the pair $(\gamma^\pi, \psi^\pi)$ from the safe set $\mathcal{C}$ and the stochastic closed-loop dynamics model.
The algorithm is then extended to an iterative scheme that alternates between updating the SCBF estimate and the backup policy, which we call the \textit{power-policy iteration}.

\subsection{Power Iteration with Backup Policy Fixed}

\begin{algorithm}[t]
\caption{Power Iteration (Policy Fixed)} \label{algo: power iteration, policy fixed}
\begin{algorithmic}[1]
\Require Initial guess $0 \leq \psi_0 \in C_0(\Omega)$ with $\norm{\psi_0} = r_0 = 1$, fixed Lipschitz backup policy $\pi:\mathcal{C}\rightarrow U$
\State $k \leftarrow 0$
\While{not converged}
    \State Compute $T_t^\pi \psi_k$ by solving \autoref{eq: T pde}.
    \State $r_{k+1} \leftarrow \norm{T_t^\pi \psi_k}$, $\gamma_{k+1} \leftarrow -\frac{1}{t} \log r_{k+1}$
    \State $\psi_{k+1} \leftarrow T_t^\pi \psi_k / r_{k+1}$
    \State $k \leftarrow k + 1$
\EndWhile
\State \Return $(\gamma_{k}, \psi_{k})$
\end{algorithmic}
\end{algorithm}

Power iteration, also known as the von Mises iteration, is an algorithm that finds the dominant eigenpair of a matrix or a linear transformation \cite{mises1929praktische}.
This algorithm is described in \autoref{algo: power iteration, policy fixed} where the power iteration is applied to the operator $T_t^\pi$.
Starting from a normalized nonnegative initial guess $\psi_0 \in C_0(\Omega)$, the method repeatedly applies the operator $T_t^\pi$ to the current iterate by propagating the PDE \autoref{eq: T pde} for horizon $t$, re-scale to normalize the result, and updates the eigenvalue estimate.
The iteration continues until convergence, after which we also obtain $\gamma^\pi = -\frac{1}{t} \log r_t^\pi$ using \autoref{thm: r semigroup}.
In the following, we show that the \autoref{algo: power iteration, policy fixed} is guaranteed to converge linearly with the convergence rate determined by the spectral gap, under the standing assumptions.

\begin{theorem} \label{thm: power iteration convergence}
    \autoref{algo: power iteration, policy fixed} converges to the unique dominant eigenpair $(r_t^\pi, \psi^\pi)$ for any normalized nonnegative initial guess $0 \leq \psi_0 \in C_0(\Omega)$ such that $\norm{\psi_0} = 1$.
    The rate of convergence is exponential, i.e.,
    \begin{equation} \label{eq: spectral convergence rates}
        \limsup_{k \rightarrow \infty} \norm{\psi_k - \psi^\pi}^{1/k} < 1, \quad \limsup_{k \rightarrow \infty} |r_k - r_t^\pi|^{1/k} < 1.
    \end{equation}
\end{theorem}
\begin{proof}
For brevity, write $T = T_t^\pi$, $r = r_t^\pi$, $\psi = \psi^\pi$, $\varphi = \varphi_t^\pi$, where $\psi$ and $\varphi$ are normalized as per \autoref{eq: eigenfunction normalization}.
First, observe that
\begin{equation}
    \psi_k = \frac{T^k \psi_0}{\norm{T^k \psi_0}}.
\end{equation}
Following a similar procedure used in the proof of \autoref{thm: prob approximation}, consider the decomposition $T = rP + N$ with $P\beta = \varphi(\beta) \cdot \psi$ ($\beta \in C_0(\Omega)$), so that
\begin{equation}
    T^k \psi_0 = (rP + N)^k \psi_0 = r^k \cdot c\cdot \psi + N^k \psi_0,
\end{equation}
with $c = \varphi(\psi_0)$. Note that $c > 0$ because $\psi_0$ is picked nonzero nonnegative.
Applying Gelfand's formula to $N$ gives $\norm{N^k} = o(\rho^k)$ for any $\rho > \boldsymbol{r}(N)$, and thus
\begin{equation}
    \psi_k = \frac{T^k \psi_0}{\norm{T^k \psi_0}} = \psi + o((\rho / r)^k).
\end{equation}
Since $\boldsymbol{r}(N) < r$ as discussed in \autoref{cor: spectral gap}, we can pick such a $\rho$ to be smaller than $r$.
This proves the convergence rate claim for the dominant eigenfunction.
On the other hand, since $r_{k + 1} = \norm{T \psi_k}$, we have $|r_{k+1} - r| \leq \norm{T}\norm{\psi_k - \psi}$. Thus, $r_k \rightarrow r$ with the same convergence rate bound.
\end{proof}

\autoref{thm: power iteration convergence} indicates that if one employs a convergence criterion in the form $\norm{\psi_k - \psi_{k-1}} \leq \epsilon_\mathrm{tol}$, the number of iterations required scales with $O(\log(1 / \epsilon_\mathrm{tol}))$ as $\epsilon_\mathrm{tol} \searrow 0$, where the implicit constant depends on the rates in \autoref{eq: spectral convergence rates}.
While the complexity of the PDE solver needed to evaluate $T_t^{\pi_{k+1}} \psi_k$ typically scales exponentially with respect to the dimensionality $n_x$ and polynomially with respect to the number of grid points per coordinate, the overall computational cost is primarily determined by the PDE solver and spatial discretization.
Convergence of the power iteration is generally faster when the initial guess is closer to the actual eigenfunction. In practice, one can \textit{warm-start} from a good initial guess, e.g., a handcrafted guess or a converged solution from a coarser grid, in order to reduce the number iterations needed to reach a given convergence threshold.

A standard caveat of power iteration is the sign ambiguity: \autoref{algo: power iteration, policy fixed} cannot distinguish between $\psi^\pi$ and $-\psi^\pi$ and can converge to either. However, since $T_t^\pi$ is a positive operator, if the iteration starts from a nonnegative initial condition $\psi$ such that $\psi(x) \geq 0$ for all $x \in \mathcal{C}$, since neither applying $T_t^\pi$ nor the normalization step can flip the sign, all subsequent iterates remain nonnegative and therefore the converged $\psi^\pi$ will be also on the nonnegative side.

\subsection{Power-Policy Iteration: Power Iteration with Policy Updates}

\begin{algorithm}[t]
\caption{Power-Policy Iteration} \label{algo: power-policy iteration}
\begin{algorithmic}[1]
\Require Initial guess $0\leq\psi_0\in C_0(\Omega) \cap C^2 (\Omega)$ with $\norm{\psi_0}=1 = r_0 = 1$,
initial backup policy $\pi_0 \in \Pi$
\State $k \leftarrow 0$
\While{not converged}
\State Select $\pi_{k+1} \in \Pi$ satisfying \autoref{eq: policy improvement}.
\State $r_{k+1} \leftarrow \norm{T_t^{\pi_{k+1}} \psi_k}$, $\gamma_{k+1} \leftarrow -\frac{1}{t} \log r_{k+1}$
\State $\psi_{k+1} \leftarrow T_t^{\pi_{k+1}} \psi_k / r_{k+1}$
\State $k \leftarrow k + 1$
\EndWhile
\State \Return $(\pi_{k}, \gamma_k, \psi_k)$
\end{algorithmic}
\end{algorithm}

The power iteration developed above computes the dominant eigenpair associated with a fixed backup policy $\pi$.
While the resulting $\psi^\pi$ is indeed a valid SCBF as seen in \autoref{thm: dominant eigenfunction is an scbf}, its probabilistic interpretation and the guarantees obtained from it are relative to $\pi$.
Consequently, if $\pi$ is an unfavorable backup policy inducing a large eigenvalue decay rate $\gamma^\pi$, the safety probability lower bound provided by the SCBF $\psi^\pi$ also decays rapidly, reflecting this unfavorable behavior.
Moreover, as we observe in \autoref{sec: asymptotic interpretation}, $\pi$ does not permit SCBF decay rates smaller than $\gamma^\pi$.
In order to make a tighter guarantee, a policy that gives a sufficiently small $\gamma^\pi$ is required.
This motivates extending power iteration to search for an improved backup policy while updating its SCBF estimate.
The resulting \textit{power-policy iteration}, described in \autoref{algo: power-policy iteration}, alternates between a policy improvement step and the power iteration step.

Let $\Pi$ denote the set of admissible policies satisfying the regularity assumptions made so far.
Given the current iterate $\psi_k$ and policy $\pi_k$, the policy update selects an admissible policy $\pi_{k+1}$ satisfying
\begin{equation} \label{eq: policy improvement}
    \mathcal{A}^{\pi_{k+1}} \psi_k (x) \geq \mathcal{A}^{\pi_{k}} \psi_k(x), \quad \forall x \in \Omega.
\end{equation}
Under the current assumptions, as shown in \autoref{lemma: semigroup regularity}, $\psi_k = T_t^{\pi_k} \psi_{k-1} / r_{k}$ is guaranteed to be twice differentiable, so one can compare the derivatives in \autoref{eq: policy improvement}.
Note that this update is feasible, because $\pi_{k+1} = \pi_k$ satisfies \autoref{eq: policy improvement}.
In the numerical implementation, a spatially discretized approximation of this update can be obtained by evaluating the pointwise maximizer at each point of a finite Cartesian grid and extending it over $\Omega$ by multilinear interpolation, i.e.,
\begin{equation} \label{eq: next policy}
    \pi_\mathrm{target}(x) \in \arg \max_{u \in U} \mathcal{A}^u \psi_k(x)
\end{equation}
for every grid point $x$.
For convex $U$, Lipschitz continuity of $\pi_{k+1}$ is given by the nonzero grid spacing and the boundedness of $U$.

\begin{figure*}[t]
    \centering
    \includegraphics[page=1, width=0.495\linewidth]{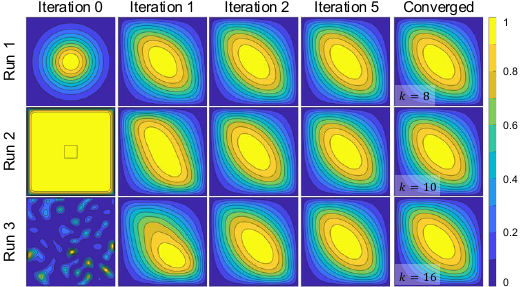}
    \includegraphics[page=2, width=0.495\linewidth]{figures/doubleint_convergence.pdf}
    \caption{
    Visualizations of $\psi_k$ (left) and $\pi_k$ (right) during the course of convergence of power-policy iteration in the stochastic double integrator example. In all runs, power-policy iteration terminated within 16 steps.
    In each subfigure, the horizontal and vertical axes represent the position and the velocity, respectively.
    }
    \label{fig: doubleint convergence}
\end{figure*}

This update rule is motivated by the following policy-improvement property.
Suppose an admissible policy $\pi' \in \Pi$ satisfies 
\begin{equation} \label{eq: policy improvement 2}
    \mathcal{A}^{\pi'} \psi^\pi (x) \geq \mathcal{A}^\pi \psi^\pi(x) = -\gamma^\pi \psi^\pi (x), \quad \forall x \in \Omega,
\end{equation}
and consider the process
\begin{equation}
    Y_s \coloneqq e^{\gamma^\pi s} \psi^\pi (X_s) \cdot 1_{\{\tau_\Omega > s\}},
\end{equation}
where $X_{(\cdot)}$ is the closed-loop state trajectory starting at $X_0 = x$ under policy $\pi'$. 
Using the inequality \autoref{eq: policy improvement 2}, one can show that $Y_s$ is a submartingale.
Consequently,
\begin{equation}
    Y_0 = \psi^\pi (x) = e^{\gamma^\pi t} T_t^\pi \psi^\pi(x) 
    \leq \mathbb{E}[Y_t] = e^{\gamma^\pi t} T^{\pi'} \psi (x),
\end{equation}
which should hold for all $x \in \Omega$.
That is, we have $T_t^{\pi'} \psi^\pi \geq T_t^\pi \psi^\pi$, and thus
\begin{equation}
    \varphi_t^{\pi'}(T_t^{\pi'} \psi^\pi) = r_t^{\pi'} \varphi_t^{\pi'} (\psi^\pi) \geq r_t^\pi \varphi_t^{\pi'} (\psi^\pi),
\end{equation}
yielding $r_t^{\pi'} \geq r_t^\pi$ and $\gamma^{\pi'} \leq \gamma^\pi$.
Following the interpretation presented in \autoref{sec: asymptotic interpretation}, the new policy $\pi'$ has an asymptotic safety probability decay rate that is not faster than $\pi$.
Thus, once $\psi_k$ has converged to the exact dominant eigenfunction, then the update rule \autoref{eq: next policy} would improve the decay rate.

Indeed, the convergence of power-policy iteration is only heuristic without a theoretical guarantee.
Since $\psi_k$ is still an intermediate estimate of $\psi^{\pi_k}$, the above argument does not by itself guarantee monotonic improvement of the true decay rate $\gamma^{(\cdot)}$ at every iteration, especially while $\psi_k$ is in the initial progress of convergence.
Nevertheless, as we report in the following numerical examples, once the initial transient phase has passed, the iterates consistently show satisfactory convergence in practice, with the decay rates improving from its initial value.
We also note that the convergence guarantee can be recovered by freezing the policy update after a finite number of iterations, after which the algorithm reduces to the original fixed-policy power iteration.

\section{Numerical Examples and Discussions} \label{sec: numerical examples}

This section presents the numerical examples.
In all examples, we use the level set toolbox \cite{mitchell2007toolbox} and the associated helperOC software\footnote{\url{https://github.com/HJReachability/helperOC}. Accessed 2026-10-21.} to solve the PDE \autoref{eq: T pde} that is required in the power-policy iteration step.

\subsection{Stochastic Double Integrator} \label{sec: stochastic double integrator}

As the first and simplest numerical example, we consider the following double integrator system under noise:
\begin{equation} \label{eq: doubleint dynamics}
    dX_t = \left(\begin{bmatrix}
        0 & 1 \\
        0 & 0
    \end{bmatrix} X_t + \begin{bmatrix}
        0 \\ 1
    \end{bmatrix} u_t \right) dt + dW_t,
\end{equation}
where the first and second components of the state vector $X_t$ are position and velocity, respectively, $u_t$ is the acceleration input, and $dW_t$ is a two-channel Brownian disturbance.
The system is assumed to be constrained by the box-shaped state and input constraints
\begin{equation}
    X_t \in \Omega = (-1, 1)\times (-2, 2), \quad 
    u_t \in U = [-1, 1].
\end{equation}
To compute $T_t^\pi$, the PDE solver was run over a horizon $t = 0.5$ with integration subintervals dynamically chosen by the PDE solver, whereas the eigenfunction estimates and the policy were represented over a uniform grid of size $101 \times 101$.
Policy updates were frozen after 12 iterations, and fixed-policy power iteration continued until the error $\norm{\psi_k - \psi_{k-1}}$ fell below $\epsilon_\mathrm{tol} = 10^{-7}$.

\begin{figure}[t]
    \centering
    \includegraphics[page=1, width=\linewidth]{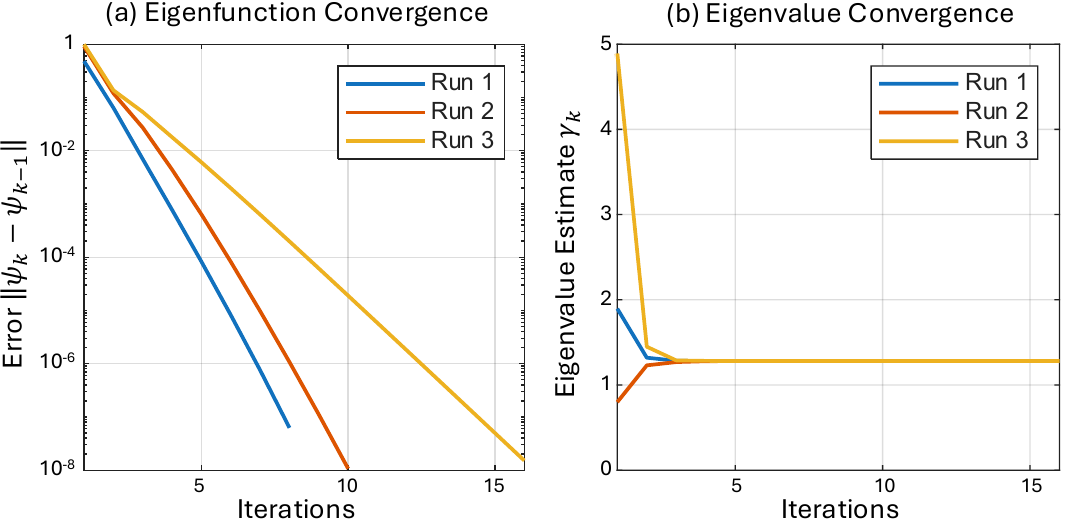}
    \caption{
    (a) The error between consecutive iterates $\psi$ during the course of convergence. The error is with respect to the Banach space norm.
    (b) The convergence of $\gamma^\pi$ value estimates. It can be seen that the iteration converges to the same value $\gamma^\pi = 1.2811$, regardless of the initial guess.
    }
    \label{fig: doubleint e and gamma}
\end{figure}

\begin{figure}[t]
    \centering
    \includegraphics[width=\linewidth, page=1]{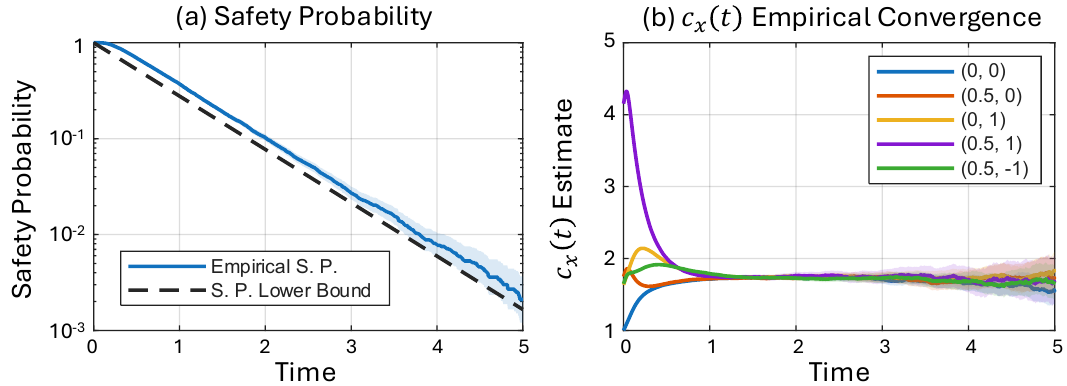}
    \caption{(a) Empirical safety probability from the initial state $X_0 = [0, 0]^\top$ and the computed SCBF-based lower bound for the stochastic double integrator system. (b) Rescaled empirical safety probabilities $c_x(t)$ for five different initial states as indicated in the legend. In both subfigures, shaded regions indicate 99\%-confidence Wilson score intervals \cite{wilson1927probable}. For the empirical estimations, $10^4$ and $5 \times 10^5$ Monte Carlo samples are used for (a) and (b), respectively. S.P. = Safety probability}
    \label{fig: doubleint monte carlo}
\end{figure}

To demonstrate the empirical robustness of convergence of the proposed power-policy iteration, we present the convergence characteristics starting from three different initializations in \autoref{fig: doubleint convergence} and \autoref{fig: doubleint e and gamma}.
As is clearly observed, all runs converged to the same computed decay rate ($\gamma^\pi = 1.2811$) and eigenfunction. 
The estimated safety probability $Z^\pi$, i.e., the fraction of simulated trajectories that have not escaped from the safe set, is shown in \autoref{fig: doubleint monte carlo}~(a), together with the bound provided by \autoref{thm: scbf prob bound}.
To empirically verify the asymptotic interpretation in \autoref{thm: prob approximation}, we also examine
\begin{equation} \label{eq: c_x}
    c_x(t) \coloneqq \frac{e^{\gamma^\pi t} Z^\pi(t,x)}{\psi^\pi (x)}
\end{equation}
for various initial conditions.
Due to \autoref{thm: prob approximation}, $c_x(t)$ values are expected to converge to the same value $c^\pi \geq 1$ regardless of $x \in \Omega$. The results shown in \autoref{fig: doubleint monte carlo}~(b) are consistent with this expectation, although these are only statistical estimates.

Finally, we compare the safety of the computed backup policy with two SOS-based baselines: the $c$-martingale approach \cite{steinhardt2012finite} and the generalized stochastic barrier approach \cite{santoyo2021barrier}.
Both baselines jointly compute a polynomial barrier function and a backup policy.
\autoref{fig: doubleint baseline comparisons}~(a) compares empirical safety probabilities under the generated backup policies.
\autoref{fig: doubleint baseline comparisons}~(b) additionally compares the probability bounds provided by these methods and the bound $\psi^\pi(x) e^{-\gamma^\pi t}$ computed by fixed-policy power iteration.
In this comparison, power-policy iteration produced a safer backup policy, both in empirical safety probability and $\gamma^\pi$ value, compared to the SOS-based methods.
Moreover, the eigenpair-based probability bound appeared to be tighter than those provided by these methods.a

\begin{figure}[t]
    \centering
    \includegraphics[page=1, width=\linewidth]{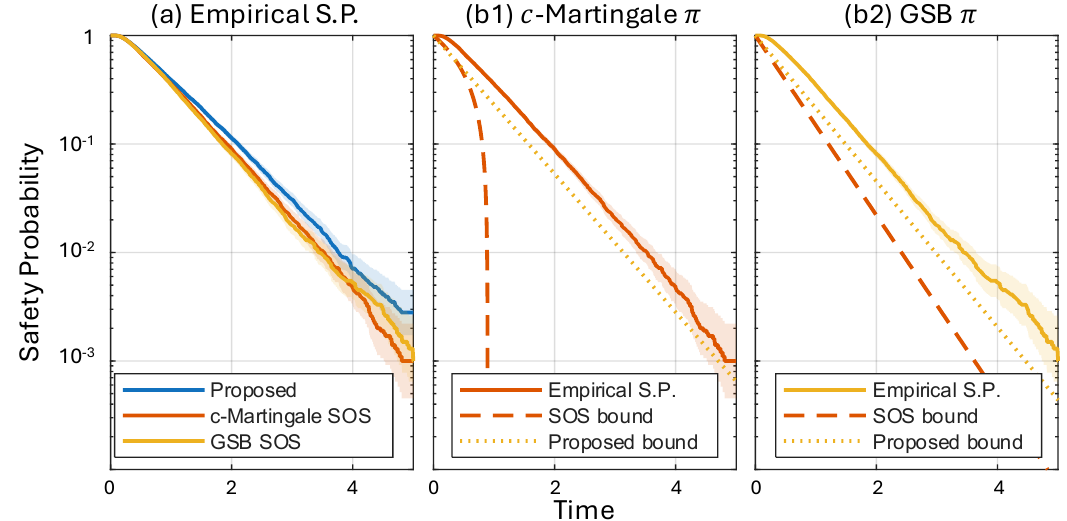}
    \caption{Baseline comparison results. We compare with SOS-based bounds \cite{steinhardt2012finite, santoyo2021barrier}. (a) Empirical safety probabilities for the three backup policies. (b) Empirical safety probability for SOS-based backup policies, with probability bounds provided by the baseline methods and fixed-policy power iteration. Shaded regions indicate 99\%-confidence Wilson score intervals \cite{wilson1927probable}. S.P. = Safety probability, GSB = Generalized stochastic barrier approach}
    \label{fig: doubleint baseline comparisons}
\end{figure}

\subsection{Wing-in-Ground-Effect Aircraft} \label{sec: aircraft}

Next, we look into a more practical example of a lightweight wing-in-ground-effect (WIG) aircraft.\footnote{A WIG is a fixed-wing aircraft that enhances efficiency by operating at low altitudes. It exploits ground effect, an aerodynamic phenomenon occurring close to the ground surface that compresses the air beneath the wing, reducing induced drag and increasing lift.}
We model its dynamics as a three-dimensional reduced-order model:
\begin{equation}
\begin{aligned}
    dH_t &= V_t \sin \Gamma_t\; dt + \Sigma_H dW^0_t \\
    dV_t &= \frac{\left( F_t \cos \alpha_t - D_t - mg \sin \Gamma_t \right)dt + \Sigma_D dW^1_t}{m} \\
    d\Gamma_t &= \frac{\left(F_t \sin \alpha_t + L_t - mg \cos \Gamma_t \right) dt + \Sigma_L dW^2_t}{m V_t},
\end{aligned}
\end{equation}
where the three components of the state vector $X_t \coloneqq [H_t, V_t, \Gamma_t]^\top$ represent the altitude, forward velocity, and flight path angle of the aircraft, respectively. 
The constants $m$ and $g$ represent the aircraft's mass and the gravitational acceleration, respectively. 
The control input $u_t = [\alpha_t, F^c_t]^\top$ of this model consists of the angle of attack $\alpha_t$ and the thrust command $F^c_t$. 
We model the propulsion efficiency drop at high airspeed \cite{goyal2024benchmarking} using the relation
\begin{equation}
    F_t = F_t^c \cdot \operatorname{sat} (1-c_F (V_t - V_F), [0, 1]),
\end{equation}
where $c_F$ and $V_F$ are positive constants and the $\operatorname{sat}$ function clips off the excessive amount of the first argument to fit within the specified interval.
While the aircraft experiences and makes use of the non-negligible ground effect, the lift $L_t$ and the drag $D_t$ are modeled as follows:
\begin{equation}
    L_t = \frac{1}{2} \rho S V_t^2 C_L(\alpha_t, H_t/b),\; D_t = \frac{1}{2} \rho S V_t^2 C_D(\alpha_t, H_t/b),
\end{equation}
with $\rho$, $S$, $b$ respectively being the air density, wing area, and wingspan; and
\begin{equation}
\begin{aligned}
    C_L(\alpha, h/b) &= C_{L,\infty}(\alpha) \cdot (1 + c_\mathrm{GE} e^{-k_L h/b}),\\
    C_D(\alpha, h/b) &= C_{D0} + k_i \cdot C_{L,\infty}(\alpha)^2 \cdot (1-e^{-k_Dh/b}),
\end{aligned}
\end{equation}
where $C_{L,\infty}(\alpha) = C_{L0} + C_{L\alpha} \alpha$ is the lift coefficient without ground effect. 
The values of the aerodynamic coefficients ($C_{L0}$, $C_{L\alpha}$, $c_\mathrm{GE}$, $k_L$, $C_{D0}$, $k_i$, $k_D$) are fitted to reasonable values based on \cite{phillips2013lifting, mantle2016induced}.
Three-channel Brownian noise $dW_t = [dW_t^0, dW_t^1, dW_t^2]^\top$ enters the model as uncertainties on the vertical velocity, lift, and drag, where their magnitudes are represented by $\Sigma_{(\cdot)}$ values.
We summarize the parameter values in \autoref{tab: aircraft_parameter_values}.

\begin{table}[t]
    \centering
    \caption{Parameter Values Used in the WIG Example}
    \begin{tabular}{c|c||c|c||c|c}
        \hline
        Name & Value & Name & Value & Name & Value \\
        \hline
        $m$ & $500\;\mathrm{kg}$ & $V_F$ & $20\;\mathrm{m/s}$ & $C_{D0}$ & $0.03$ \\
        $g$ & $9.81\;\mathrm{m/s^2}$ & $\rho$ & $1.225\;\mathrm{kg/m^3}$ & $c_\mathrm{GE}$ & $0.2$ \\
        $\Sigma_H^2$ & $1\;\mathrm{m^2/s}$ & $S$ & $12\;\mathrm{m^2}$ & $k_L$ & $5$ \\
        $\Sigma_L^2$ & $6\times10^4\;\mathrm{N^2/s}$ & $b$ & $10\;\mathrm{m}$ & $k_D$ & $5$ \\
        $\Sigma_D^2$ & $156\;\mathrm{N^2/s}$ & $C_{L0}$ & $0.2$ & $k_i$ & $0.05$ \\
        $c_F$ & $0.02\;\mathrm{s/m}$ & $C_{L\alpha}$ & $5\;\mathrm{/rad}$ & & \\
        \hline
    \end{tabular}
    \label{tab: aircraft_parameter_values}
\end{table}

Since WIGs fly at a low altitude, the primary safety concern is to maintain a nonzero clearance from the ground. 
Upon that, we also impose maximum altitude, airspeed, and flight angle constraints so that the vehicle does not excessively deviate from the nominal flight condition. Thus, the set $\Omega$ is given as a box constraint set as follows:
\begin{equation} \label{eq: aircraft constraints}
\begin{aligned}
    0\;\mathrm{m} &<  H_t < 10 \; \mathrm{m}, \\
    20 \;\mathrm{m/s} &< V_t < 70 \; \mathrm{m/s},\\
    -0.15 \; \mathrm{rad} &< \Gamma_t < 0.15 \; \mathrm{rad}.
\end{aligned}
\end{equation}
We also set upper and lower limits to the thrust command and angle of attack inputs as follows:
\begin{equation}
    0 \; \mathrm{rad} \leq \alpha_t \leq 0.2 \; \mathrm{rad},  \quad
    0 \; \mathrm{N} \leq F_t^c \leq 1000 \; \mathrm{N}.
\end{equation}

The input fed into the system is given as the solution to the SCBF-QP \autoref{eq: scbf-qp} with decay rate $\gamma = 0.005\;\mathrm{s^{-1}}$, and the reference input is set as a constant
\begin{equation}
    \pi_\mathrm{ref}(t,x) = [0.03 \; \mathrm{rad},\; 300 \; \mathrm{N}]^\top.
\end{equation}
While the used dynamics model is not affine in input, the filtered input signals coming from the nonlinear optimization \autoref{eq: scbf-qp} are not guaranteed to be globally optimal.
However, as long as they are feasible solutions to \autoref{eq: scbf-qp}, the probabilistic safety guarantee given by \autoref{thm: scbf prob bound} remains valid.

\begin{figure}[t]
    \centering
    \includegraphics[page=1, width=\linewidth]{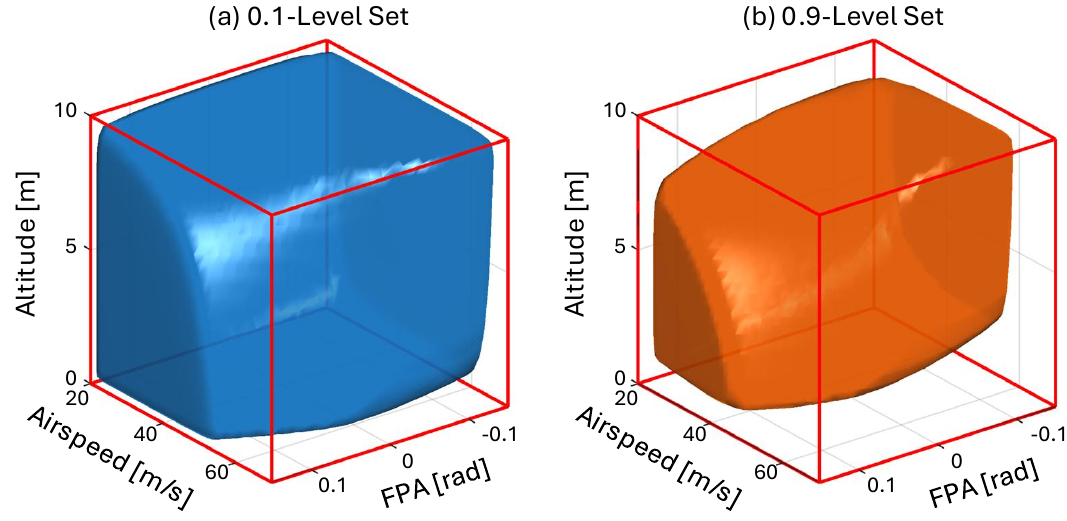}
    \caption{The level sets of the computed eigenfunction in the WIG example. The red skeletons denote the safety constraints \autoref{eq: aircraft constraints}. FPA  = Flight path angle.}
    \label{fig: wig eigenfunction}
\end{figure}
\begin{figure}[t]
    \centering
    \includegraphics[page=1, width=\linewidth]{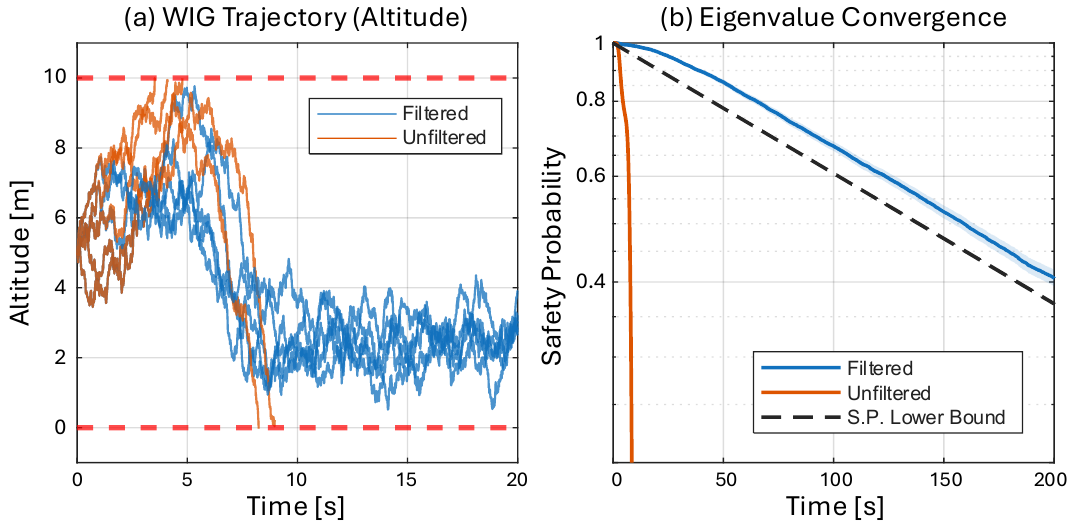}
    \caption{(a) Closed-loop trajectories (altitude) of filtered and unfiltered trajectories for the WIG example. Red dashed lines indicate the safety constraint in \autoref{eq: aircraft constraints}. The trajectories start at $X_0 = [5\;\mathrm{m}, 45\;\mathrm{m/s}, 0\;\mathrm{rad}]^\top$.
    (b) Monte Carlo safety probability estimations over 200-second horizon for both policies using $10^4$ samples. Shaded region indicates 99\%-confidence Wilson score intervals \cite{wilson1927probable}. 
    The safety filter ensures that the system's safety probability is maintained above the lower bound provided by \autoref{thm: scbf prob bound}. 
    S.P. = Safety probability.}
    \label{fig: wig trajectories}
\end{figure}

The eigenfunction and policy were represented on a $31 \times 41 \times 31$ grid, using a 50-second power-iteration horizon.
Policy updates were frozen after 20 iterations, followed by fixed-policy power iteration with $\epsilon_\mathrm{tol} = 10^{-7}$.
The $0.5$-level set of the computed eigenfunction is drawn in \autoref{fig: wig eigenfunction}.
The $\gamma^\pi$ value converged to $3.36 \times 10^{-6} \; \mathrm{s^{-1}}$.
In \autoref{fig: wig trajectories}, we compare the safety of the filtered policy against unfiltered, through Monte Carlo simulations starting from the initial state $X_0 = [5\;\mathrm{m}, 45\;\mathrm{m/s}, 0\;\mathrm{rad}]^\top$.
It can be seen that the safety filter successfully modifies the reference policy $\pi_\mathrm{ref}$ and that the resulting closed-loop safety probability remains above the probability bound, whereas no unfiltered trials survived beyond $15\;\mathrm{s}$.

\subsection{Stochastic Cart-Pole System}

\begin{table}[t]
    \centering
    \caption{Parameter Values Used in the Cart-Pole Example}
    \begin{tabular}{c|c|c}
        \hline
        Symbol & Description & Value \\
        \hline
        $M$ & Cart mass & $1 \; \mathrm{kg}$ \\
        $m$ & Pole mass & $0.2 \; \mathrm{kg}$ \\
        $\ell$ & Pole length & $1 \; \mathrm{m}$ \\
        $g$ & Gravitational acceleration & $9.81\;\mathrm{m/s^2}$ \\
        \hline
    \end{tabular}
    \label{tab: cart pole parameters}
\end{table}

As a higher-dimensional nonlinear example, we consider a stochastic cart-pole system, consisting of a cart moving along a horizontal track and a pole attached to the cart through a frictionless pivot.
The pole is modeled as a massless rod with a point mass attached to the opposite end of the pivot.
We model the system dynamics as follows:
\begin{equation}
    dX_t = [v_t, \omega_t, a_v(X_t, u_t), a_\omega(X_t, u_t)]^\top dt + \Sigma dW_t,
\end{equation}
where the four components of the state $X_t = [p_t, \theta_t, v_t, \omega_t]^\top$ denote the cart position, pole angle, cart velocity, and pole angular velocity, respectively. 
The input $u_t$ is a one-dimensional horizontal force applied to the cart.
The cart's linear and pole's angular accelerations are given as
\begin{equation}
\begin{aligned}
    a_v(x, u) &= \frac{u + m \sin \theta (\ell \omega^2 - g \cos \theta)}{M + m \sin^2 \theta}, \\
    a_\omega(x,u) &= \frac{(M + m) g \sin \theta - m \ell \omega^2 \cos \theta \sin \theta - u \cos \theta}{\ell (M + m \sin^2 \theta)}.
\end{aligned}
\end{equation}
The meanings and the values of the physical parameters are summarized in \autoref{tab: cart pole parameters}.
The process $dW_t$ is a four-dimensional standard Brownian motion, and the diffusion matrix is set as $\Sigma = \operatorname{diag}(0.1, 0.04, 0.3, 0.3)$ in appropriate SI units.
In the simulation, we impose the safety constraint described as
\begin{equation}
\begin{aligned}
    |p_t| < 2\;\mathrm{m}, \quad |\theta_t| < 0.65\;\mathrm{rad}, \\
    |v_t| < 3 \; \mathrm{m/s}, \quad |\omega_t| < 4 \; \mathrm{rad/s}
\end{aligned}
\end{equation}
and the input constraint
\begin{equation}
    |u_t| \leq 12\;\mathrm{N}.
\end{equation}

\begin{figure}[t]
    \centering
    \includegraphics[width=\linewidth]{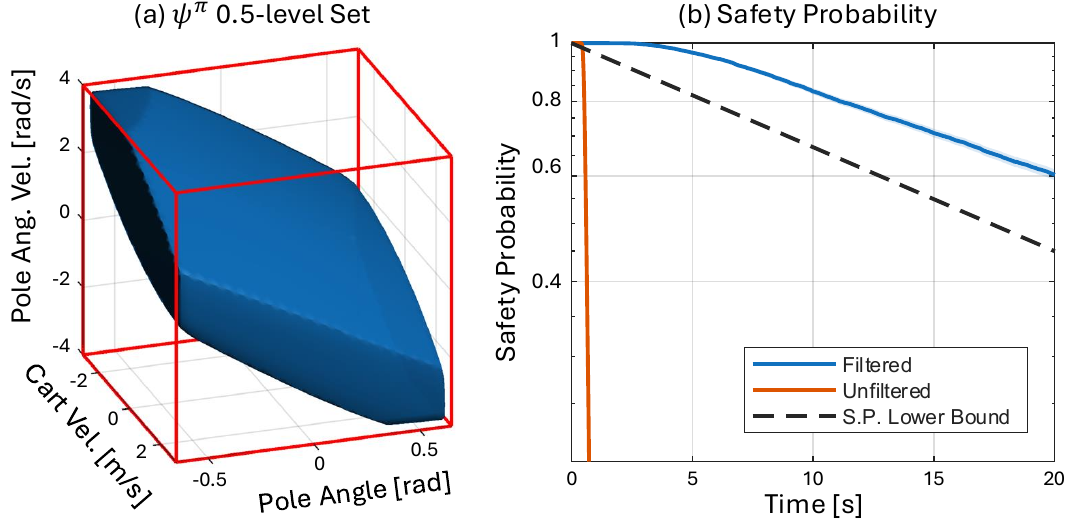}
    \caption{Simulation results for the stochastic cart-pole example. (a) The computed dominant eigenfunction $\psi^\pi$. The $0.5$-level set of the three-dimensional slice with position fixed to $0\;\mathrm{m}$ is shown. (b) Monte Carlo safety probability estimates ($10^4$ samples) from the initial state $X_0 = [0\;\mathrm{m}, 0\;\mathrm{rad}, 0\;\mathrm{m/s}, 0\;\mathrm{rad/s}]^\top$ under SCBF-QP with decay rate $\gamma = 0.04\;\mathrm{s^{-1}}$, and for the unfiltered reference policy. Shaded region indicates 99\%-confidence Wilson score intervals \cite{wilson1927probable}. It can be seen that the safety probability under SCBF-QP stays above the lower bound provided by \autoref{thm: scbf prob bound}.
    Ang. Vel. = Angular velocity, Vel. = Velocity, S.P. = Safety probability.}
    \label{fig: cart pole}
\end{figure}

Computation used a $41^4$ grid and power-policy iteration with 0.5-second horizon.
Policy updates were frozen at 20 iterations, followed by fixed-policy power iteration with $\epsilon_{\mathrm{tol}} = 10^{-7}$.
This yielded $\gamma^\pi = 0.0285 \; \mathrm{s^{-1}}$.
\autoref{fig: cart pole}~(a) shows the 0.5-level set of the converged eigenfunction on the $p_t = 0$ slice.
Safety filter based on SCBF-QP was implemented using $\gamma = 0.04 \; \mathrm{s^{-1}}$ and the reference policy
\begin{equation}
    \pi_\mathrm{ref}(x) = \operatorname{sat}(-K x, [-12\;\mathrm{N}, 12\;\mathrm{N}])
\end{equation}
with $K = [1, 20, 1, 5]$, each entry in the appropriate SI unit.
This reference policy was chosen to provide poor stabilization, so that safety is primarily provided by the safety filter.
In \autoref{fig: cart pole}~(b), we compare and plot the safety probabilities for safety-filtered and unfiltered policies, with the probability lower bound $\psi^\pi (x) e^{-\gamma t}$ provided for the safety-filtered policy.

\subsection{Empirical Extension to Rank-Deficient Diffusion} \label{sec: rank deficient}

The formal results developed so far relied on \autoref{assumption: full rank diffusion} which stated that the diffusion term $\sigma(x,u)$ should have full row rank for all $x$ and $u$.
This makes the differential operator $\mathcal{A}^\pi$ uniformly elliptic, and this enabled the use of established elliptic operator theory to derive useful properties of the operators and their eigenmodes.
However, there exist systems for which a rank-deficient diffusion ($\sigma(\cdot, \cdot)$ not having full row rank), thereby violating \autoref{assumption: full rank diffusion}, provides a more appropriate model.
Still, some basic properties underlying the framework remain true regardless of whether \autoref{assumption: full rank diffusion} is satisfied.
In particular, the operator $\mathcal{T}_t^\pi$ remains a positive linear operator satisfying the semigroup property.
On the other hand, the additional properties established using elliptic operator theory are no longer guaranteed.
This includes the regularity and positivity-improving properties, compactness, simplicity of the dominant eigenvalue, a strictly positive dominant eigenfunction, and a nonzero spectral gap. 
Nevertheless, we empirically observe that the spectral structure and power-policy iteration can exhibit qualitatively similar behavior even when \autoref{assumption: full rank diffusion} is violated.

Consider a variant of stochastic double integrator system discussed in \autoref{sec: stochastic double integrator}:
\begin{equation}
    dX_t = \left(\begin{bmatrix}
        0 & 1 \\
        0 & 0
    \end{bmatrix} X_t + \begin{bmatrix}
        0 \\ 1
    \end{bmatrix} u_t \right) dt + \begin{bmatrix}
        0 \\ 1
    \end{bmatrix} dW_t,
\end{equation}
under the same state and input constraints $\Omega = (-1, 1) \times (-2, 2)$ and $U = [-1, 1]$. Unlike \autoref{eq: doubleint dynamics}, $dW_t$ is a one-dimensional standard Brownian motion.
This is an obvious violation of \autoref{assumption: full rank diffusion}.
However, power-policy iteration succeeded to converge to a fixed point, which is visualized in \autoref{fig: rank deficient}~(a).
\autoref{fig: rank deficient}~(b) is the result of a similar experiment to \autoref{sec: stochastic double integrator} using the $c_x(t)$ values defined in \autoref{eq: c_x}.
The results show a similar trend to the observations in \autoref{fig: doubleint monte carlo}.

Even if power-policy iteration empirically converged in this particular case, it is not guaranteed to be a valid SCBF since its regularity and boundary behaviors do not necessarily follow what is established under \autoref{assumption: full rank diffusion}.
Still, this observation suggests that the asymptotic spectral characterization may remain valid for certain rank-deficient systems, motivating further theoretical investigation.

\begin{figure}[t]
    \centering
    \includegraphics[page=1, width=\linewidth]{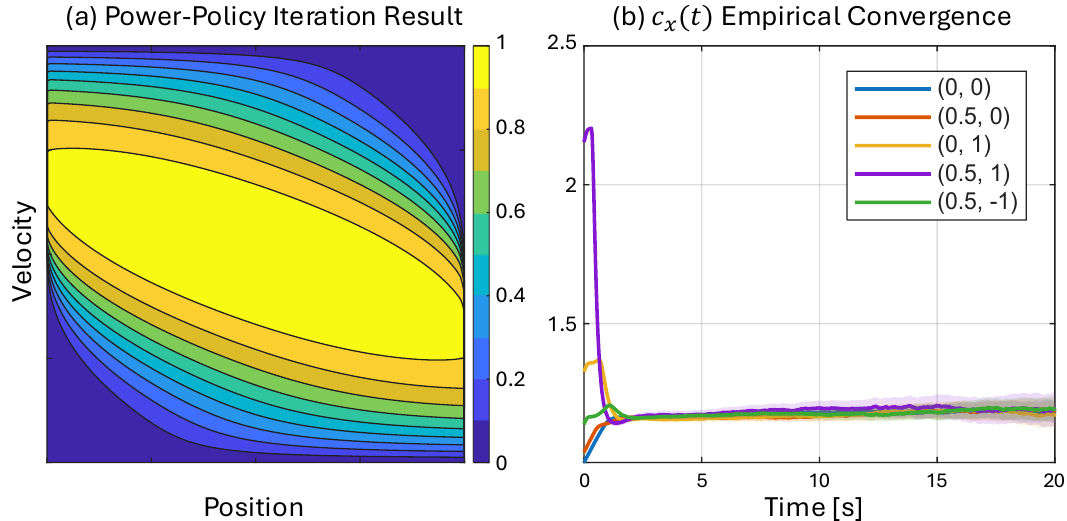}
    \caption{Power-policy iteration produced these results for the stochastic double integrator variant considered in \autoref{sec: rank deficient}.
    Although outside the guarantees made in this paper, the properties similar to those observed in the full-rank cases are observed in the same manner in this rank-deficient case.
    The converged $\psi^\pi$ function is not necessarily a valid SCBF, as the required regularity and boundary conditions are not guaranteed.
    The $c_x(t)$ values are estimated through Monte Carlo simulation with $10^5$ samples. Shaded regions indicate 99\%-confidence Wilson score intervals \cite{wilson1927probable}.}
    \label{fig: rank deficient}
\end{figure}

\section{Conclusion} \label{sec: conclusion}

This paper presents an operator-theoretic framework for computing a valid and meaningful SCBF.
By analyzing the Koopman-like operator governing the evolution of the system's safety probability, we demonstrate that its dominant eigenfunction constitutes a natural and valid SCBF for the system.
In addition to providing an SCBF-based lower-bound guarantee on the safety probability, the values of this dominant eigenfunction characterize the asymptotic evolution of safety probability of a given state, thereby offering an interpretable safety metric.
The corresponding eigenvalue, called the dominant eigenvalue, describes the global rate at which the safety probability decays and serves as an indicator for the overall safety of the closed-loop system. We propose the power-policy iteration algorithm which jointly estimates the dominant eigenfunction and updates the backup policy with the goal of improving long-term safety.
The validity of the proposed method is demonstrated through simulation experiments.

Future research directions are suggested as follows.
First, while the dominant eigenpair captures asymptotic safety behavior, it offers limited insight into short-term safety. 
Future investigations should explore non-dominant eigenmodes, which can be complex, to provide tighter safety probability estimates over finite horizons.
Second, while the current framework relies on full state observability with infinite bandwidth, extending this to partially observable systems by integrating it with state estimation techniques remains a key objective.
Finally, acknowledging that this method is subject to the curse of dimensionality, a challenge common in PDE-based methods like HJ reachability analysis, future efforts should focus on developing scalable approximation techniques, such as neural-network-based representations \cite{bansal2021deepreach}, for higher-dimensional systems.

\section*{Acknowledgment}
The authors would like to thank Prof. Jason J. Choi (UCLA), Prof. Namhoon Cho (Seoul National University), Prof. Somil Bansal (Stanford University), Dylan Hirsch (UC San Diego), Jonghae Park (Seoul National University), and Sihyun Cho (Seoul National University) for insightful discussions.

\bibliography{references}

\begin{IEEEbiography}[{\includegraphics[width=1in,height=1.25in,clip,keepaspectratio]{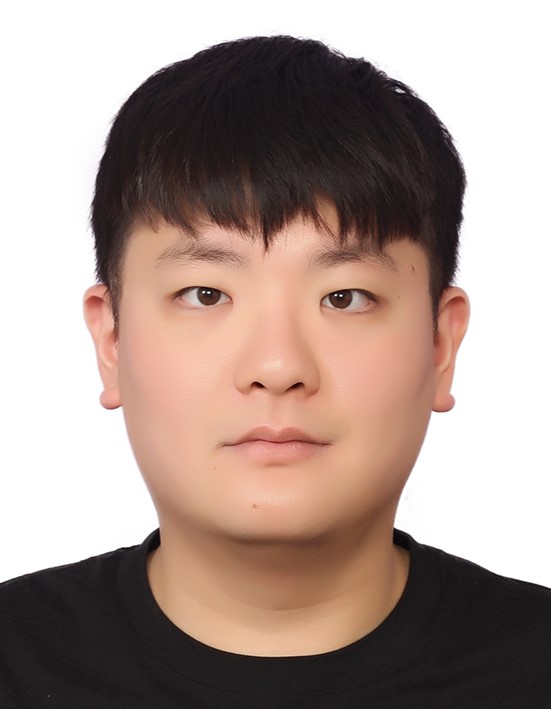}}]{Inkyu Jang} (Member, IEEE) received the B.S. degree in Mechanical Engineering in 2020 and the Ph.D. degree in Aerospace Engineering in 2026, both from Seoul National University, Seoul, South Korea.
He is currently a postdoctoral researcher with the Department of Electrical Engineering and Computer Sciences, University of California, Berkeley, CA, USA.
His research interests include safety-critical and stochastic control, particularly their connections to machine learning and applications in robotics.
\end{IEEEbiography}

\begin{IEEEbiography}[{\includegraphics[width=1in,clip,keepaspectratio]{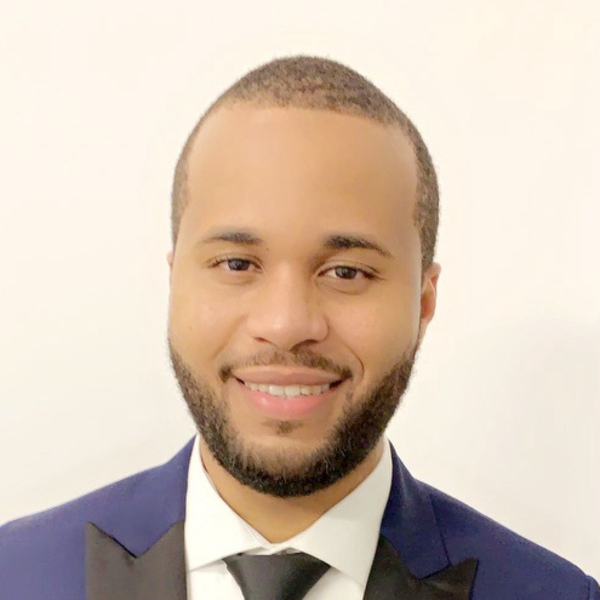}}]{Chams E. Mballo} received the Ph.D. in Aerospace Engineering (2022) and an M.S. in Mathematics (2021) from the Georgia Institute of Technology, Atlanta. He is currently a Postdoctoral Fellow at the University of California, Berkeley. His research focuses on safety-critical control, reachability analysis, human–machine interaction, and sustainable aviation for next-generation electric aircraft.
\end{IEEEbiography}

\begin{IEEEbiography}[{\includegraphics[width=1in,clip,keepaspectratio]{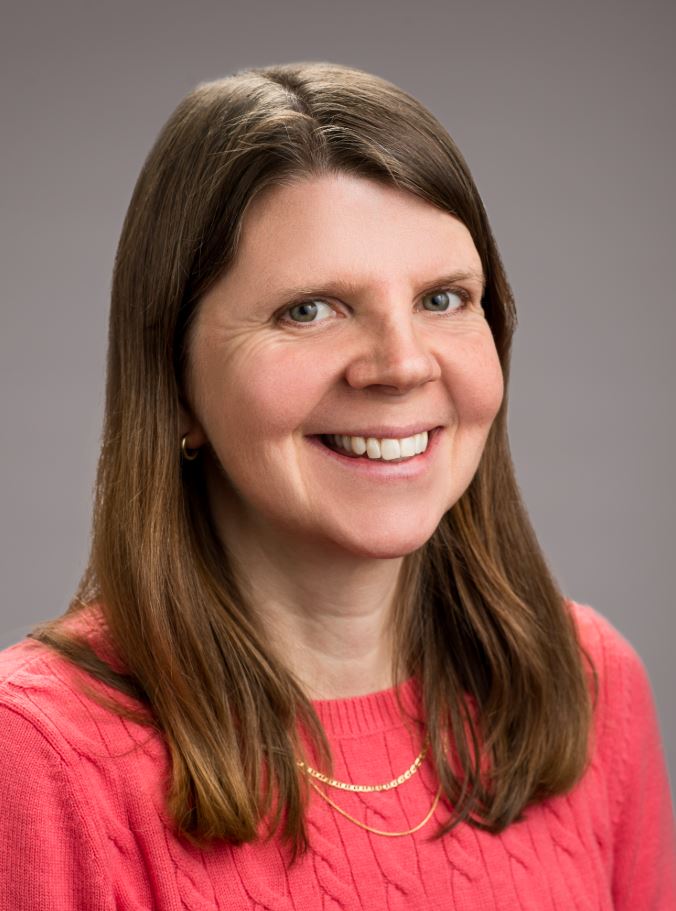}}]{Claire J. Tomlin}  (Fellow, IEEE) received the Ph.D. degree in electrical engineering and computer science from the University of California, Berkeley, CA, USA, in 1998. She is the James and Katherine Lau Professor of Engineering and the Professor and Chair with the Department of Electrical Engineering and Computer Sciences at UC Berkeley, Berkeley, CA, USA. From 1998 to 2007, she was an Assistant, Associate, and Full Professor in Aeronautics and Astronautics at Stanford University, Stanford, CA. In 2005, she joined UC Berkeley. Her research interests include control theory and hybrid systems, with applications to air traffic management, UAV systems, energy, robotics, and systems biology.

Prof. Tomlin is a MacArthur Foundation Fellow (2006). She was the recipient of the IEEE Transportation Technologies Award in 2017. In 2019, she was elected to the National Academy of Engineering and the American Academy of Arts and Sciences.
\end{IEEEbiography}

\begin{IEEEbiography}[{\includegraphics[width=1in,clip,keepaspectratio]{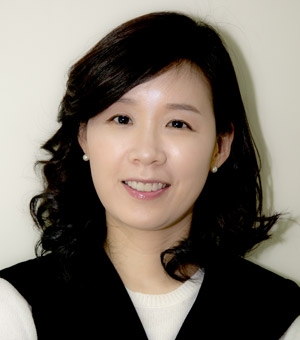}}]{H. Jin Kim} (Member, IEEE) received the B.S. degree from Korea Advanced Institute of Technology, Daejeon, Korea, in 1995, and the M.S. and Ph.D. degrees in mechanical engineering from the University of California at Berkeley (UC Berkeley), Berkeley, CA, USA, in 1999 and 2001, respectively. 

From 2002 to 2004, she was a Postdoctoral Researcher with the Department of Electrical Engineering and Computer Science, UC Berkeley. In September 2004, she joined the Department of Mechanical and Aerospace Engineering, Seoul National University, Seoul, Korea, as an Assistant Professor, where she is currently a Professor. Her research interests include intelligent control of robotic systems and motion planning.
\end{IEEEbiography}

\end{document}